\documentclass{amsart}

\usepackage[babel,threshold=2]{csquotes}
\usepackage{amsfonts}
\usepackage{amssymb}
\usepackage{enumerate}
\usepackage{amsmath}
\usepackage{amsthm}
\usepackage{graphicx}
\usepackage[dvipsnames]{xcolor}
\usepackage[english]{babel}
\usepackage[doi=false,backend=biber,style=numeric,sorting=none]{biblatex}
\usepackage{comment}
\usepackage{hyperref}

\numberwithin{equation}{section}

\newtheorem{theorem}{Theorem}[section]
\newtheorem{lemma}[theorem]{Lemma}
\newtheorem{claim}[theorem]{Claim}

\newtheorem{problem}[theorem]{Problem}

\newtheorem{corollary}[theorem]{Corollary}

\newtheorem{fact}[theorem]{Fact}
\newtheorem*{BHCS}{Theorem~\ref{Th:BHCS}}
\newtheorem*{odd}{Theorem~\ref{t:odd}}
\newtheorem*{order}{Theorem~\ref{t:order}}
\newtheorem*{I0I}{Claim~\ref{cl:I0I}}
\newtheorem*{T:equiv}{Theorem~\ref{t:equiv}}
\newtheorem*{C:prec}{Claim~\ref{cl:prec}}

\theoremstyle{definition}
\newtheorem{example}[theorem]{Example}
\newtheorem{remark}[theorem]{Remark}
\newtheorem{definition}[theorem]{Definition}
\newtheorem{notation}[theorem]{Notation}

\DeclareMathOperator{\Tr}{Tr}
\DeclareMathOperator{\Sp}{Sp}
\DeclareMathOperator{\PG}{PG}

\newcommand{\NN}{\mathbb{N}}

\newcommand{\RR}{\mathbb{R}}

\newcommand{\CC}{\mathbb{C}}

\newcommand{\iG}{\mathcal{G}}
\newcommand{\iH}{\mathcal{H}}
\newcommand{\iI}{\mathcal{I}}

\newcommand{\iL}{\mathcal{L}}

\newcommand{\iF}{\mathcal{F}}
\newcommand{\iP}{\mathcal{P}}
\newcommand{\iS}{\mathcal{S}}

\newcommand{\iZ}{\mathcal{Z}}

\newcommand*{\defeq}{\stackrel{\text{def}}{=}}

\begin{document}
	
\title[Positivity and entanglement of polynomial Gaussian operators]{Positivity and entanglement of polynomial Gaussian integral operators}
\author{Rich\'ard Balka} 
\address{HUN-REN Alfr\'ed R\'enyi Institute of Mathematics, Re\'altanoda u.~13--15, H-1053 Budapest, Hungary, AND Institute of Mathematics and Informatics, Eszterh\'azy K\'aroly Catholic University, Le\'anyka u.~4, H-3300 Eger, Hungary}
\email{balkaricsi@gmail.com}

\author{Andr\'as Csord\'as}
\address{Department of Physics of Complex Systems, E\"{o}tv\"{o}s Lor\'and University, P\'azm\'any P\'eter s\'et\'any 1/A, H-1117 Budapest, Hungary}
\email{csordas@tristan.elte.hu}

\author{G\'abor Homa}
\address{HUN-REN Wigner Research Centre for Physics, P.~O.~Box 49, H-1525 Budapest, Hungary}
\email{homa.gabor@wigner.hun-ren.hu}



\begin{abstract} Positivity preservation is an important issue in the dynamics of open quantum systems: positivity violations always mark the border of validity of the model. We investigate the positivity of self-adjoint \emph{polynomial Gaussian} integral operators $\widehat{\kappa}_{\PG}$, that is, the multivariable kernel $\kappa_{\PG}$ is a product of a polynomial $P$ and a Gaussian kernel $\kappa_G$. These operators frequently appear in open quantum systems. 
 
We show that $\widehat{\kappa}_{\PG}$ can be only positive if the Gaussian part is positive, which yields a strong and quite easy test for positivity. This has an important corollary for the bipartite entanglement of the density operators $\widehat{\kappa}_{\PG}$: if the Gaussian density operator $\widehat{\kappa}_G$ fails the Peres--Horodecki criterion, then the corresponding polynomial Gaussian density operators $\widehat{\kappa}_{\PG}$ also fail the criterion for all $P$, hence they are all entangled. 

We prove that polynomial Gaussian operators with polynomials of odd degree cannot be positive semidefinite.

We introduce a new preorder $\preceq$ on Gaussian kernels such that if $\kappa_{G_0}\preceq \kappa_{G_1}$ then $\widehat{\kappa}_{\PG_0}\geq 0$ implies $\widehat{\kappa}_{\PG_1}\geq 0$ for all polynomials $P$. Therefore, deciding the positivity of a polynomial Gaussian operator determines the positivity of a lot of another polynomial Gaussian operators having the same polynomial factor, which might improve any given positivity test by carrying it out on a much larger set of operators. We will show an example that this really can make positivity tests much more sensitive and efficient. This preorder has implication for the entanglement problem, too.

\end{abstract}

\maketitle

\section{Introduction}

\subsection{History}

The concept of self-adjoint integral operators has been in the center of mathematical and physical research over the last century.
Fredholm, Hilbert, Mercer, and Schmidt were the first ones who established the cornerstones of this area \cite{Fredholm1903,Hilbert1904,Mercer,schmidt1905entwickelung,Schmidt1907,Stewart}, which is still an active field of research, see the monographs \cite{Arvind1995,Folland+1989,Gosson_Harmonic,deGosson_book,Gosson_trace_class,Gosson}. The foundation of quantum mechanics in the $1920$s brought even more attention to the notions of integral operators and infinite dimensional Hilbert spaces. One of the most important mathematical objects of quantum information theory and quantum physics is the density operator, which is a positive semidefinite, self-adjoint, trace-class operator with trace one \cite{Neumann}.

Technically, density operators acting on infinite dimensional Hilbert spaces can be represented by $L^2$ kernels in position representation, or equivalently by Wigner functions in phase space, or with characteristic functions in another phase space \cite{Serafini,nielsen_chuang_2010}. Phase space methods were first applied by Weyl, Wigner, Husimi, and Moyal \cite{Weyl,Wigner,Husimi,Moyal}, and they found many applications in mathematics, quantum chemistry, statistical mechanics, quantum optics and quantum information theory \cite{Folland+1989,Hillery, Lee, Schleich, Weinbub,Salazar,Adesso}.  

Despite the large literature on integral operators, it is still a very hard problem to determine the spectrum, and even the positivity of operators in concrete cases. The positivity check of different models of quantum mechanics has already started in the $1960$s with the so-called KLM conditions \cite{Kastler, Loupias1, Loupias2}, and later further studies of trace-class operators were carried out \cite{Gosson_trace_class,Narkowich1, Narkowich2, Werner, Luef}. It is still an active area of research to check the positivity of an operator without actually calculating the full spectrum \cite{Gosson,Newton}. The dynamics of the density operator of a subsystem are usually given by master equations, where evolution in time is governed by partial differential equations. Usually, a quantum-mechanical system interacts with external quantum systems and these interactions significantly change the dynamics of the examined subsystem, causing quantum dissipation and decoherence \cite{book1}. A good example for this phenomenon is the quantum Brownian motion \cite{HPZ,H-Y}, where the master equations are derived from first principles and they may lead to positivity violations of the density operator at a later time depending on the external parameters of the model \cite{Gnutzmann,BLH,HCSB,HBCSCS}.

At the beginning of modern quantum mechanics a new phenomenon, the quantum entanglement and its complementary notion, the separability emerged \cite{EPR,Bell}, and they still pose serious challenges to both mathematicians and theoretical physicists \cite{Open,Zyczkowski}. In two breakthrough papers, Peres \cite{Perescrit} and the Horodecki family \cite{HORODECKI19961} independently found a necessary condition for bipartite separability by checking positivity after the so-called partial transpose operation, which provides a very important link between entanglement and positivity. Another important result is due to Werner and Wolf, who proved (based on a previous work of Simon) that the Peres--Horodecki criterion is a necessary and sufficient condition of separability for all bipartite $1$ \textit{versus} $n$ modes Gaussian quantum states \cite{WernerWolf,SimonR,Duan,Lami}. Determining entanglement of non-Gaussian density operators in infinite dimensional Hilbert spaces is an extremely difficult task \cite{Pirandola,Vogel,Miki,Volume}. Another approach to the entanglement problem is using so-called witness operators given by G\"uhne and T\'oth \cite{GUHNE}. The situation becomes even more complicated if we consider the time evolution of the entanglement of composite quantum systems \cite{nielsen_chuang_2010,kina_entanglement}.

\subsection{Motivation and physical applications} \label{ss:motivation}

Gaussian kernels occur quite frequently in physics, for example coupled harmonic oscillators in thermal distribution can be described by Gaussian operators of the form 
\begin{equation*} \widehat{\rho}=\frac{\exp{\big(-\beta \widehat{H}_{\text{osc}}\big)}}{\Tr\big[\exp{\big(-\beta \widehat{H}_{\text{osc}}\big)}\big]}, 
\end{equation*}
if the Hamiltonian operator $\widehat{H}_{\text{osc}}$is at most quadratic in position and momentum operators \cite{Plenio_2004}. The positivity of a Gaussian kernel is fully determined by its covariance matrix (see \cite{Serafini}), while the behaviour of non-Gaussian kernels is more subtle. Our main goal is to consider more general \emph{polynomial Gaussian} integral operators, that is, operators whose kernels can be written as a polynomial multiple of a Gaussian, see \eqref{eq:polynomial Gaussian_kernel_def} for the precise definition. Contrary to the Gaussian case, for polynomial Gaussian operators no finite positivity test is known. 

In recent years, the development of quantum technology  pointed out that working with non-Gaussian quantum states is essential in applied quantum informatics, see \cite{Wu,Yanagimoto} and the references therein. Specifically, polynomial Gaussian forms appear naturally in physics: papers  use this form at the intersection of various areas such as theoretical  and experimental atomic, molecular and quantum optical physics \cite{Lukas,Szabo,meng2023entangled,Roux_polygaussian,Lvovsky,Walschaers}, quantum optomechanics  \cite{Lukas,Lvovsky,Walschaers}, cavity quantum electrodynamics systems \cite{Walschaers,Elliott}, physics of nanomechanical resonators \cite{Remus} and their quantum technological applications \cite{Lukas,meng2023entangled,Lvovsky,Walschaers}. Another appearance of polynomial Gaussian forms is excited states of coupled oscillators in position representation: polynomial Gaussian operators can be obtained from those excited states by mixing some of them with positive weights \cite{kina_entanglement,Remus,Dattoli1996,PhysRevA.54.5378}. 
Studying the entanglement of the above composite system with the Peres--Horodecki criterion translates to checking the positivity of a polynomial Gaussian operator. Our strategy is to prove results about the positivity of polynomial Gaussian operators, which can be translated to the language of entanglement by the Peres--Horodecki criterion. The entanglement problems have already found several applications in various disciplines e.\,g.~quantum chemistry \cite{Ding2021}, quantum information theory \cite{Coladangelo2020,Infinite}, quantum optics \cite{meng2023entangled}, quantum communication \cite{Piveteau2022,Bodri,Koniorczyk2024}  quantum computers \cite{Fauseweh2024}, and quantum gravity \cite{Fahn_2023, hsiang2024graviton}. A recent mathematical breakthrough \cite{MIPRE,MIPRE2} in the area of quantum computational complexity provides connection to the areas of quantum information, operator algebras, and approximate representation theory; in particular it solves Tsirelson's problem \cite{Tsirelson_Bell} from quantum information theory, and the equivalent Connes' embedding problem \cite{Connes} from the theory of von Neumann algebras. For the above problems see also \cite{Junge}, where it is stated that ``finite dimensional quantum models do not suffice to describe all bipartite correlations'', emphasizing the role of operators acting on infinite dimensional Hilbert spaces. This wide applicability and connections between different areas served as a source of inspiration for us, too. 

\subsubsection{Positivity violations of the density operator in the Caldeira--Leggett master equation}
\label{subsec:Positivity}
The Caldeira--Leggett master equation \cite{BLH} for a quantum harmonic oscillator with  
frequency $\omega$ and $\hbar=m=k_B=1$ is 
\begin{eqnarray*}
&&i \frac{\partial}{\partial t} \rho(x,y,t)=\Bigg[ \frac{1 }{2}\left(\frac{\partial^2}{\partial y^2}-
\frac{\partial^2}{\partial x^2} \right) +\frac{ \omega^2}{2} \left(x^2-y^2\right) 
-i D_{pp} (x-y)^2 
\\
&& -i \gamma (x-y) \left(\frac{\partial}{\partial x}-
\frac{\partial}{\partial y} \right) -2 D_{px} (x-y) \left(\frac{\partial}{\partial x}+
\frac{\partial}{\partial y} \right) \Bigg]  \rho(x,y,t), 
\end{eqnarray*}
where $\gamma$ denotes the coupling constant, and $D_{pp}$ and $D_{px}$ stand for the momentum diffusion coefficient and the cross diffusion coefficient, respectively. This master equation describes the time evolution of the density operator of the central harmonic oscillator, where an environment of harmonic oscillators in thermal equilibrium with a sufficiently high  temperature $T$ is considered with Ohmic spectral density and a high frequency cut-off $\Omega$. If the temperature is not high enough, in certain time intervals the solution of the equation is not physical, that is, the kernel $\rho$ is not positive semidefinite in those intervals. From the physical point of view the solution loses its validity at the very first time when $\widehat{\rho}$ is not positive semidefinite.

In \cite{BLH} different initial eigenstates of the quantum harmonic oscillator ($n=0,1,2$) are considered, namely
\begin{equation} \label{hos}
\rho_n(R,r,0)= 
\sqrt {{\frac {{
\beta}^{2}}{\pi }}}\frac {H_n \left(\beta R+\frac{\beta r}{2} \right) H_n\left(\beta R-\frac{\beta r}{2}\right)}{{2}^{n}n!} e^{-\frac{1}{2}\,{\beta}^{2}
\left( 2\,{R}^{2}+\frac{1}{2}\,{r}^{2} \right) },
\end{equation}
where we parametrized by the center of mass and relative coordinates 
$R=\frac{x+y}{2}$ and $r=x-y$, respectively. Here $H_n$ denotes the $n$th Hermite polynomial and the parameter $\beta$ is the inverse width, that is, the reciprocal of the width $x_0=\sqrt{\hbar / (m \omega)}$ corresponding to the quantum harmonic oscillator's ground state. 

In all three cases the solution of the master equation is a polynomial Gaussian operator of degree $2n$ at any finite time $t$.   
The time evolution of two important necessary conditions for the positivity of the operator are monitored in \cite{BLH}: at each time we need $\Tr{\hat{\rho}^2} \leq 1$, and also the Robertson--Schr\"odinger uncertainty relation (for its definition see \cite{BLH} and the references therein).
\begin{equation} \label{SRI}
\sigma_{RS} \geqslant \frac{1}{4}\left|\langle \hat{x}  \hat{p}-\hat{p}  \hat{x} \rangle \right|^2, 
\end{equation}
where 
\begin{equation*}
\sigma_{RS}=\Delta \hat{x} ^2 \Delta \hat{p} ^2 -\left( \frac{\langle  \hat{x}  \hat{p}+\hat{p}  \hat{x} \rangle}{2}-\langle  \hat{x} \rangle\langle   \hat{p}\rangle\right)^2.
\end{equation*}

It has been found in \cite{BLH,Erratum2019} that in all the studied cases the positivity violations in the Gaussian case ($n=0$) happened in a larger time set than in the polynomial Gaussian cases $(n=1,2)$ for both tests; note that the three kernels have the same Gaussian part at any finite time. For the time evolution of the Robertson--Schr\"odinger uncertainty relation \eqref{SRI} see Fig.~\ref{Fig: Lisztes}, more details can be found in \cite{BLH,Erratum2019}. The figure suggests that a non-positive semidefinite Gaussian kernel multiplied by a polynomial of high enough (even) degree might result in a positive semidefinite operator. We prove in Theorem~\ref{Th:BHCS} that this can never happen, and hence we obtain a strong and fairly easy positivity test as well.  

\subsubsection{Entanglement between two oscillators in the presence of heat bath}
A new toy model \cite{kina_entanglement} to study is given by two oscillators coupled to each other harmonically and interacting with the harmonic oscillator bath. Here a possible initial state can be a polynomial Gaussian operator, which remains polynomial Gaussian during the time evolution. This paper gives tools to handle the entanglement problem of this model. In Subsection~\ref{ss:entanglement} we derive the useful sufficient condition Corollary~\ref{c:NPT}, which often allows us to consider only the Gaussian part of the operator. Also, our preorder defined in Subsection~\ref{ss:preorder} makes testing positivity more sensitive and efficient, see the example at the end of that subsection. Generalizations for more than two harmonically coupled harmonic oscillators or two harmonically coupled three-dimensional harmonic oscillators are straightforward.

\subsubsection{Positivity checks and entanglement detection after applying quantum operations} In quantum information the operations are the basic tools to
modify the system under study. Our results for the positivity check can be applied to the modified operators as well, provided the quantum channel maps the polynomial Gaussian density operator to another polynomial Gaussian operator. Simplest examples are partial trace or partial transpose operations used for entanglement detection. Some other channels, defined by linear transformations in the phase space keep the polynomial Gaussian form, including symplectic transformations. A thorough discussion of such operations can be found in the book \cite{Serafini}.

\begin{figure}
 \begin{center}
  \includegraphics[width=\textwidth]{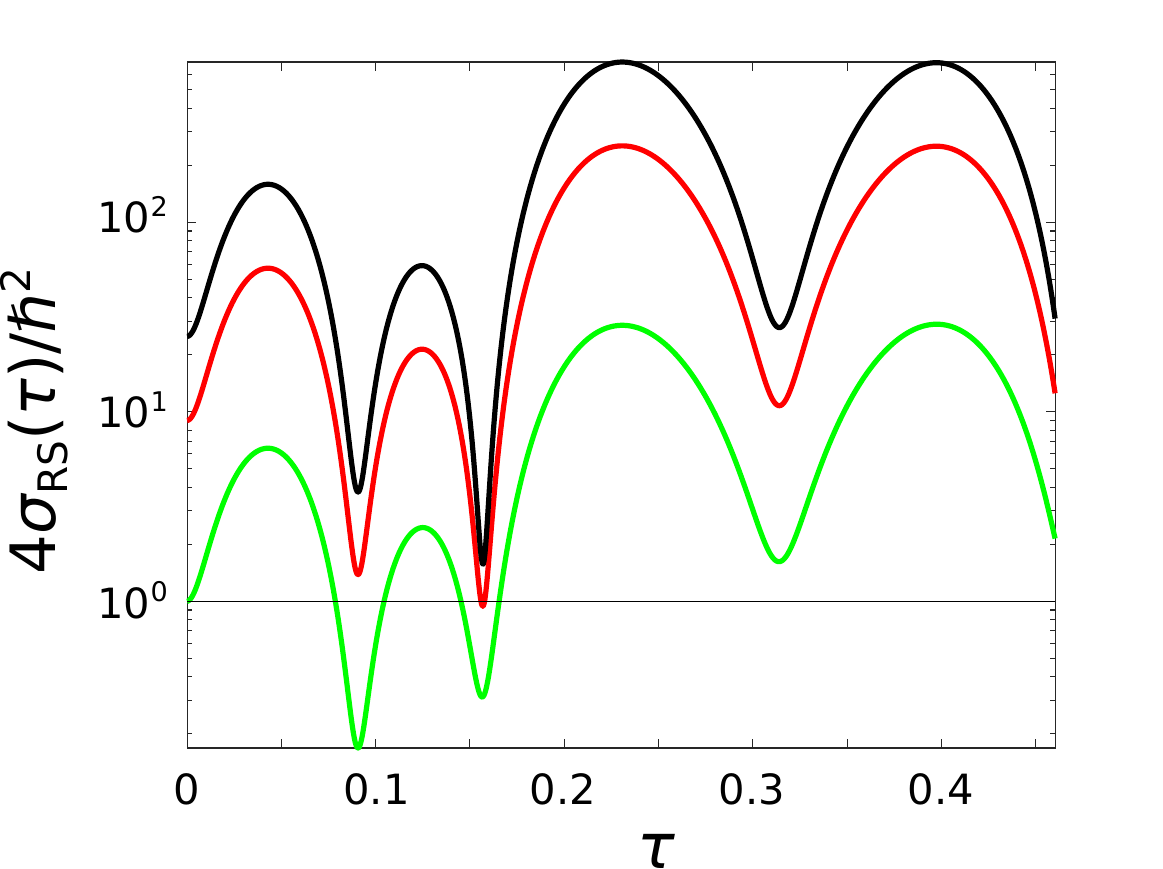}
  \end{center}
  \caption{Plot of $4\sigma_{RS}(\tau)/\hbar^2$ as a function of a dimensionless time parameter $\tau=\omega t$ on a semi-logarithmic scale coming from the Robertson--Schr\"odinger uncertainty relation \eqref{SRI}. The parameters and the initial states are the same as in Figure\,1(b) of \cite{Erratum2019}, and also $\hbar$ is reinstated. The three different curves are plotted for the initial states given in \eqref{hos}: $n=0$ (green); $n=1$ (red); and $n=2$ (black). If a curve goes below $1$, positivity is violated.}
  \label{Fig: Lisztes}
 \end{figure}

\subsection{Theoretical background}

We denote by $L^2\left(\mathbb{R}^{n}\right)$ the Hilbert space of the complex-valued square integrable functions defined on $\RR^n$ with the scalar product\footnote{In physics the usual convention is $\langle f,g \rangle=\int_{\RR^n} f^{*}(x)g(x) \, \mathrm{d} x$.} $\langle f,g \rangle=\int_{\RR^n} f(x)g^{*}(x) \, \mathrm{d} x$, where $z^{*}$ denotes the complex conjugate of $z$. A kernel $\kappa \in L^2(\mathbb{R}^{2n})$ defines an integral operator $\widehat{\kappa}\colon L^2\left(\mathbb{R}^{n}\right)\to L^2\left(\mathbb{R}^{n}\right)$ by the formula
\begin{equation*}
 \left ( \widehat{\kappa} f \right) (x) = \int_{\mathbb{R}^{n}}\, \kappa(x,y) f(y)\, \mathrm{d}y, 
\end{equation*}
see e.\,g.~\cite{deGosson_book,Stein}. We call $\widehat{\kappa}$ \emph{self-adjoint} if $\widehat{\kappa}=\widehat{\kappa}^\dagger$. For continuous kernels $\kappa$ this is equivalent to the property that $\kappa(y,x)=\kappa^*(x,y)$ for all $x,y\in \RR^n$. We say that $\widehat{\kappa}$ is \emph{positive semidefinite} if $\langle \widehat{\kappa} f, f\rangle \geq 0$ for all $f\in L^2(\RR^n)$; we also use the notation $\widehat{\kappa}\geq 0$. If $\widehat{\kappa}$ is positive semidefinite, it is necessarily self-adjoint. As $\widehat{\kappa}$ is a \emph{Hilbert--Schmidt} operator, it is \emph{compact} \cite[Proposition 9.12]{deGosson_book}, see also \cite{deGosson_book} for the definitions. From now on assume that $\widehat{\kappa}$ is self-adjoint. As $\widehat{\kappa}$ is compact and self-adjoint, the spectral theorem \cite[Theorem~3 in Chapter~28]{lax2002functional} yields that it has countably many eigenvalues $\{\lambda_i\}^\infty_{i=0}$, and there is an orthonormal basis consisting of eigenvectors. The eigenvalue equation of $\widehat{\kappa}$ is a Fredholm-type integral equation
\begin{equation*}
\int_{\RR^n} \, \kappa(x,y) \phi_i(y)\,\mathrm{d}y=\lambda_i \phi_i(x), \quad \text{where }  \phi_i\in L^2(\mathbb{R}^n). 
\end{equation*} 
We say that $\widehat{\kappa}$ is \emph{trace-class}\footnote{Linear operators satisfy the inclusion: $\text{trace-class} \subset \text{Hilbert--Schmidt}\subset \text{compact} \subset \text{bounded}$.} if $\sum_{i=0}^{\infty} |\lambda_i| < \infty$, and we define the \emph{trace} of $\widehat{\kappa}$ as
\begin{equation*}
 \Tr (\widehat{\kappa})\defeq \sum_{i=0}^{\infty} \lambda_i. 
\end{equation*}
If $\kappa$ is continuous, then (see \cite{Brislawn} and the original \cite{Duflo}) we have the formula 
\begin{equation} \label{eq:trace}
\Tr (\widehat{\kappa})=\int_{\RR^n} \kappa(x,x) \, \mathrm{d}x.
\end{equation} 
The \emph{Schwartz space} $\iS(\RR^n)$ is the set of rapidly decreasing smooth functions defined as follows. A smooth function $f\colon \RR^n\to \CC$ satisfies $f\in \iS(\RR^n)$ if for all multiindices $\alpha=(\alpha_1,\dots,\alpha_n)\in \NN^n$ and $\beta=(\beta_1,\dots,\beta_n)\in \NN^n$ we have 
\begin{equation*}
\sup_{x\in \RR^n} |x^{\alpha} (D^{\beta} f)(x)|<\infty, 
\end{equation*} 
where we use the notation $x^{\alpha}=x_1^{\alpha_1}\cdots x_n^{\alpha_n}$ and $D^{\beta}=\partial_{1}^{\beta_1}\cdots \partial_{n}^{\beta_n}$; for more on Schwartz functions see e.\,g.~\cite[Section~V.3]{SR}. If $\kappa\in \iS(\RR^{2n})$ is a Schwartz kernel, then $\widehat{\kappa}$ is a trace-class operator, see \cite[Proposition~287]{Gosson_Harmonic} or \cite[Proposition~1.1]{Brislawn} with the remark afterwards. We say that $\widehat{\rho}$ is a \emph{density operator} if it is positive semidefinite with $\Tr(\widehat{\rho})=1$, so its eigenvalues satisfy $\lambda_i \geq 0$ and $\sum_{i=0}^{\infty} \lambda_i=1$. For more on Hilbert--Schmidt, compact, and trace-class operators see e.\,g.~\cite{Gosson_Harmonic,deGosson_book,Stein}. The following theorem basically dates back to Mercer \cite{Mercer}.

\begin{theorem}[Mercer] \label{t:Mercer} Assume that $\kappa\in L^2(\RR^{2n})$ is a continuous kernel. Then $\widehat{\kappa}$ is positive semidefinite if and only if for all $x_1,\dots, x_k \in \RR^n$ and $c_1,\dots, c_k\in \CC$ we have 
\begin{equation} \label{eq:Mercer}  
\sum_{i,j=1}^k c_i c^{*}_j \kappa(x_i,x_j)\geq 0. 
\end{equation} 
\end{theorem}

Let $\kappa\in \iS(\RR^{2n})$ be a Schwartz kernel. For $\widehat{\kappa}$ the \emph{Wigner--Weyl transform} connects the position representation $\kappa(x,y)$ and the \emph{phase space representation}\footnote{If $\widehat{\kappa}$ is a density operator, then $W(x,p)$ is called \emph{Wigner function}.} $W(x,p)$ such that for all $x,y,p\in \RR^n$ we have
\begin{equation}
\label{eq:Wigner-transform}
W(x,p)=\frac{1}{(2 \pi)^n}\int_{\mathbb{R}^n}  \exp{\left\{ - i p^{\top} y\right\}} \kappa\left(x+\frac{y}{2},x-\frac{y}{2}\right) \, \mathrm{d} y,
\end{equation}
and the inverse transform is
\begin{equation}
 \label{eq:Inverse Wigner-transform}
\kappa(x,y)= \int_{\mathbb{R}^n} W\left(\frac{x+y}{2},p\right) 
\exp{\left\{ i p^{\top} \left( x-y  \right)\right\}} \, \mathrm{d} p;
\end{equation}
for more on these concepts the reader might consult the monograph \cite{Folland+1989}. Note that here we used the convention that the Planck-constant satisfies $\hbar=1$.

\subsection{Structure of the paper}
In Section~\ref{s:statement} we state all the results of the paper, where we also dedicate a sizeable subsection to the entanglement problem. The results of the paper are proved in the later sections. We prove Theorem~\ref{Th:BHCS} in Section~\ref{s:Gauss}. First we prove it in one dimension in Subsection~\ref{ss:one}, then we summarize the knowledge on the symplectic decomposition of Gaussian operators in Subsection~\ref{ss:W}, and finally we fully prove Theorem~\ref{Th:BHCS} in Subsection~\ref{ss:gen} by tracing back the general case to the one-dimensional result. Section~\ref{s:odd} is dedicated to the proof of Theorem~\ref{t:odd}. The results concerning our preorder will be proved in Section~\ref{s:order}, namely Theorems~\ref{t:order} and \ref{t:equiv}, Claim~\ref{cl:I0I}, and Claim~\ref{cl:prec}. We finish the paper with discussion and conclusions in Section~\ref{s:discussion}.

\section{Statement of results} \label{s:statement} 
Let $A,B,C\in \RR^{n\times n}$ such that $A$ and $C$ are symmetric and positive definite and $B$ is arbitrary. We consider the \emph{Gaussian kernel} $\kappa_{G}\colon \RR^{2n}\to \CC$ as
\begin{equation*} \label{eq:kappaG}
\kappa_{G}(x,y)=\exp\left\{-(x-y)^{\top} A (x-y) -i(x-y)^{\top} B (x+y)-(x+y)^{\top} C (x+y)\right\}. 
\end{equation*}
Let $P\colon \RR^{2n}\to \CC$ be a polynomial in $2n$ variables and define the \emph{polynomial Gaussian kernel}\footnote{The same terminology was used in \cite{Roux_polygaussian} in a similar context.} $\kappa_{\PG}$ as\footnote{We could have added the linear term $iD^{\top}(x-y)+E^{\top}(x+y)$ to the exponent of $\kappa_G$ with some given vectors $D,E\in \RR^n$. However, $D$ does not change the spectrum of $\widehat{\kappa}_{\PG}$, and $E$ does not change its positivity, so they are irrelevant to us.}
\begin{equation} 
\label{eq:polynomial Gaussian_kernel_def}
\kappa_{\PG}(x,y)=P(x,y)\kappa_{G}(x,y). 
\end{equation}
Since $\kappa_{\PG}\in \iS(\RR^{2n})$, the operator $\widehat{\kappa}_{\PG}$ is trace-class. Note that if $\widehat{\kappa}_{\PG}\geq 0$ then $P$ must be self-adjoint, that is, $P(y,x)=P^{*}(x,y)$ for all $x,y\in \RR^n$. 
\subsection{When the Gaussian is not positive}

The following theorem is the main result of our paper, stating that if a Gaussian operator is not positive semidefinite, then the corresponding polynomial Gaussians cannot be positive semidefinite either. 

\begin{theorem} \label{Th:BHCS} 
Let $P\colon \RR^{2n}\to \CC$ be a non-zero polynomial and let $\kappa_G\colon \RR^{2n}\to \CC$ be a Gaussian kernel. If $\widehat{\kappa}_G$ is not positive semidefinite, then $\widehat{\kappa}_{\PG}$ is not positive semidefinite either. 
\end{theorem}

\subsection{When the polynomial is of odd degree}

We show that if a polynomial $P$ is of odd degree, then no corresponding polynomial Gaussian operator can be positive semidefinite. 

\begin{theorem} \label{t:odd} Assume that $P\colon \RR^{2n}\to \CC$ is a polynomial of odd degree, and $\kappa_G\colon \RR^{2n}\to \CC$ is any Gaussian kernel. Then $\widehat{\kappa}_{\PG}$ is not positive semidefinite. 
\end{theorem}

The above theorem generalizes \cite[Proposition~3.1]{Newton}, where the significantly easier $n=\deg P=1$ case was settled. It also explains why polynomial Gaussian operators with polynomials of odd degree do not really appear in physics. 

Let us say that $P$ is \emph{reducible to odd degree} if there is an index set $I\subset \{1,\dots,n\}$ such that substituting $x_i=y_i=0$ for all $i\in I$ into $P$ transforms $P$ into a polynomial of odd degree in $2(n-|I|)$ variables. Theorem~\ref{t:Mercer} easily implies that applying this transformation to a kernel maps a positive semidefinite operator into a positive semidefinite one. Therefore, Theorem~\ref{t:odd} immediately yields the following generalization. 

\begin{theorem} \label{t:reducible} Let $P\colon \RR^{2n}\to \CC$ be a polynomial which is reducible to odd degree, and let $\kappa_G\colon \RR^{2n}\to \CC$ be a Gaussian kernel. Then $\widehat{\kappa}_{\PG}$ is not positive semidefinite. 
\end{theorem} 
 
\subsection{A useful preorder on Gaussian kernels} \label{ss:preorder}
We introduce a relation which allows us to vary the Gaussian part in polynomial Gaussian (and more general) kernels.  
Fix a positive integer $n$, and consider matrices $A,B,C\in \RR^{n\times n}$ such that both $A$ and $C$ are symmetric. Let us define self-adjoint Gaussian functions $\theta=\theta(A,B,C)\colon  \RR^{2n} \to \CC$ as
\begin{equation*}
\theta(x,y)=\exp\left\{-(x-y)^{\top} A (x-y)-i(x-y)^{\top} B (x+y)-(x+y)^{\top}C (x+y)\right\}. 
\end{equation*}
Clearly $\theta$ is a kernel in $L^2(\RR^{2n})$ if and only if both $A$ and $C$ are positive definite. We want to address the following problem. 

\begin{problem} Let $\sigma\colon \RR^{2n} \to \CC$ be a given self-adjoint function and assume that $\sigma \theta(A_i,B_i,C_i)\in L^2(\RR^{2n})$ for $i\in \{0,1\}$. When is it true that 
\begin{equation*} \widehat{\sigma \theta}(A_0,B_0,C_0)\geq 0 \quad \Longrightarrow \quad  \widehat{\sigma \theta}(A_1,B_1,C_1)\geq 0?
\end{equation*}
\end{problem}	
	
The following definition and theorem answer this problem completely. 

\begin{notation} Let $I_n\in \RR^{n\times n}$ denote the identity matrix. Let $\iG(n)$ be the set of triples $(A,B,C)$ such that $A,B,C\in \RR^{n\times n}$ and $A,C$ are symmetric. Let
\begin{equation*} \iG^+(n)=\left\{(A,B,C)\in \iG(n): A,C \text{ are positive definite and } \widehat{\theta}(A,B,C)\geq 0\right\}. 
\end{equation*} 
\end{notation}

\begin{definition}  \label{d:prec}
We define a relation $(\iG,\preceq)$ on $\iG(n)$ as follows. For two triples $(A_0, B_0,C_0), (A_1, B_1,C_1)\in \iG(n)$ we write $(A_0, B_0,C_0)\preceq (A_1,B_1, C_1)$ if there exists $r\geq 0$ such that 
\begin{equation*} 
(A_1-A_0+rI_n, B_1-B_0, C_1-C_0+rI_n)\in \iG^+(n).
\end{equation*} 
For two Gaussian kernels $\kappa_{G_0}$ and $\kappa_{G_1}$ we write $\kappa_{G_0}\preceq \kappa_{G_1}$ if their defining matrix triples satisfy $(A_0,B_0,C_0)\preceq (A_1,B_1,C_1)$.
\end{definition}

\begin{remark}
In Definition~\ref{d:prec} the matrices $A_1-A_0$ and $C_1-C_0$ are symmetric, so for large enough $r$ the matrices $A_1-A_0+rI_n$ and $C_1-C_0+rI_n$ are both positive definite, yielding ${\theta}(A_1-A_0+rI_n, B_1-B_0, C_1-C_0+rI_n)\in L^2(\RR^{2n})$. Theorem~\ref{t:equiv} will imply that if $(A_0,B_0,C_0)\preceq (A_1,B_1,C_1)$, then any choice of $r\geq 0$ witnesses it as long as the matrices $A_1-A_0+rI_n$ and $C_1-C_0+rI_n$ are positive definite. Hence checking the relation $\preceq$ requires to determine the positivity of a single Gaussian operator, which can be easily done by calculating the eigenvalues of a $2n\times 2n$ matrix as we will see in Claim~\ref{c:positivity}.
\end{remark}

The following theorem states that this is precisely the right definition for us. 

\begin{theorem} \label{t:order}
For $(A_0, B_0,C_0), (A_1, B_1,C_1)\in \iG(n)$ the following are equivalent: 
\begin{enumerate} 
\item  $(A_0, B_0,C_0)\preceq (A_1,B_1, C_1)$;
\item for each self-adjoint $\sigma\colon \RR^{2n}\to \CC$ if $\sigma \theta(A_i,B_i,C_i)\in L^2(\RR^{2n})$ for $i\in\{0,1\}$ and 
$\widehat{\sigma \theta}(A_0,B_0,C_0)\geq 0$, then $\widehat{\sigma \theta}(A_1,B_1,C_1)\geq 0$;
\item for every self-adjoint Gaussian kernel $\sigma$ if $\sigma \theta(A_i,B_i,C_i)\in L^2(\RR^{2n})$ for $i\in\{0,1\}$ and $\widehat{\sigma \theta}(A_0,B_0,C_0)\geq 0$, then $\widehat{\sigma \theta}(A_1,B_1,C_1)\geq 0$.
\end{enumerate} 
\end{theorem}

The next corollary states that deciding the positivity of a single operator determines the positivity of a lot of other operators as well.

\begin{corollary} \label{c:two} Let a self-adjoint function $\sigma\colon \RR^{2n}\to \CC$ and $(A_0, B_0,C_0)\in \iG(n)$ be given such that $\sigma \theta(A_0, B_0,C_0)\in L^2(\RR^{2n})$, and consider all triples $(A_1,B_1,C_1)$ for which $\sigma \theta(A_1, B_1,C_1)\in L^2(\RR^{2n})$. Assume that we have already determined the positivity of $\widehat{\sigma \theta}(A_0, B_0,C_0)$, then there are two possibilities:
\begin{enumerate}[(i)]
\item \label{item:1} $\widehat{\sigma \theta}(A_0, B_0,C_0)\geq 0 \Rightarrow  \widehat{\sigma \theta}(A_1, B_1,C_1)\geq 0 \text{ for all } (A_1,B_1,C_1)\succeq (A_0,B_0,C_0)$; 
\item $\widehat{\sigma \theta}(A_0, B_0,C_0)\not\geq 0 \Rightarrow  \widehat{\sigma \theta}(A_1, B_1,C_1)\not\geq 0 \text{ for all } (A_1,B_1,C_1)\preceq (A_0,B_0,C_0)$.
\end{enumerate}
\end{corollary}

The most important corollary of Theorem~\ref{t:order} is that we can improve positivity tests by applying them to a larger class of operators obtained by switching the Gaussian part of our kernel to Gaussian kernels larger than equal to the original one with respect to our preorder $\preceq$, see Corollary~\ref{c:two}\,\eqref{item:1}. More details will be given in  Subsubsection~\ref{sss:ex}.

Theorem~\ref{t:order} (and directly Fact~\ref{f:mu}) imply that the relation $\preceq$ is transitive, and clearly reflexive, which makes it a preorder. We say that a polynomial $P\colon \RR^n \to \CC$ is \emph{universal} if $\widehat{P\kappa}_G\geq 0$ for any Gaussian kernel $\kappa_G$ satisfying $\widehat{\kappa}_G\geq 0$. The next claim characterizes universal polynomials.

\begin{claim} \label{cl:I0I} We have $(I_n,0,I_n)\preceq (A,B,C)$ for all $(A,B,C)\in \iG^{+}(n)$. Hence for a self-adjoint polynomial $P\colon \RR^{2n}\to \CC$ the following are equivalent: 
\begin{enumerate} 
\item \label{i:Gauss} $P$ is universal, that is, $\widehat{P\theta}(A,B,C)\geq 0$ for all $(A,B,C)\in \iG^{+}(n)$;
\item  $\widehat{P\theta}(I_n,0,I_n)\geq 0$;
\item \label{i:Mercer} $\sum_{i,j=1}^k c_i c^{*}_j P(x_i,x_j)\geq 0$ for all $x_1,\dots, x_k \in \RR^n$ and $c_1,\dots, c_k\in \CC$.
\end{enumerate} 
\end{claim}

As it is much easier to check \eqref{i:Mercer} than universality, the above claim is especially useful. We demonstrate this by the following example. 

\begin{example} \label{ex} Let $P\colon \RR^2\to \RR$ defined as $P(x,y)=x^{\ell}y^m+x^my^{\ell}$ with $\ell,m\in \NN$. Then $P$ satisfies Claim~\eqref{cl:I0I}\,\eqref{i:Gauss} if and only if $\ell=m$. Indeed, if $\ell=m$ then \eqref{i:Mercer} is clearly satisfied. Assume to the contrary that $\ell>m$. Applying \eqref{i:Mercer} with $c_1=2, c_2=-1$, $x_1=x$, $x_2=y$ and using the notation $\lambda=\frac yx$ we obtain that 
\begin{equation} \label{eq:lm} 
0\leq 8x^{\ell+m}-4(x^{\ell}y^m+x^my^{\ell})+2y^{\ell+m}=2x^{\ell+m}(\lambda^{\ell}-2)(\lambda^m-2).
\end{equation}
We can choose $x,y>0$ such that $\lambda^m<2<\lambda^{\ell}$, which contradicts \eqref{eq:lm}. 
\end{example}
The following claim gives quite simple  necessary \eqref{i:easy1} and also sufficient \eqref{i:easy2} conditions for our relation. 
\begin{claim} \label{cl:prec}
Let $(A_i,B_i,C_i)\in \iG(n)$, $i\in \{0,1\}$. Then 
\begin{enumerate}
\item \label{i:easy1} $(A_0,B_0,C_0)\preceq (A_1,B_1,C_1)$ $\Rightarrow$ $A_1-C_1\geq A_0-C_0$,
\item  \label{i:easy2} $A_1-C_1\geq A_0-C_0$ and $B_1-B_0$ is symmetric $\Rightarrow$ $(A_0,B_0,C_0)\preceq (A_1,B_1,C_1)$.
\end{enumerate}
\end{claim} 

\begin{notation}
Let $(A_i,B_i,C_i)\in \iG(n)$, $(i=0,1)$. We define the equivalence relation $(A_0,B_0,C_0)\approx (A_1,B_1,C_1)$ such that 
\begin{equation*} 
(A_0,B_0,C_0)\preceq (A_1,B_1,C_1) \text{ and } (A_1,B_1,C_1)\preceq (A_0,B_0,C_0).  
\end{equation*} 
For two Gaussian kernels $\kappa_{G_0}$ and $\kappa_{G_1}$ we write $\kappa_{G_0}\approx \kappa_{G_1}$ if their defining matrix triples satisfy $(A_0,B_0,C_0)\approx (A_1,B_1,C_1)$.
\end{notation}

The next theorem characterizes the equivalence of the triples $(A,B,C)\in \iG(n)$. 

\begin{theorem} \label{t:equiv} Let $(A_i,B_i,C_i)\in \iG(n)$, $i\in \{0,1\}$. 
The following are equivalent: 
\begin{enumerate} 
\item  $(A_0,B_0,C_0)\approx (A_1,B_1,C_1)$,
\item $A_1-C_1=A_0-C_0$ and $B_1-B_0$ is symmetric. 
\end{enumerate} 
\end{theorem}

Theorem~\ref{t:equiv} implies that the characterization of Gaussian kernels $\kappa_{G_0}\approx \kappa_{G_1}$ is much simpler than of $\kappa_{G_0}\preceq \kappa_{G_1}$. Therefore, it can be effectively applied to improve positivity tests by changing the Gaussian part of a  kernel to an equivalent Gaussian, for a detailed example see the following. 

\subsubsection{Application of our preorder: an example} \label{sss:ex}

Assume that $\kappa\colon \RR^{2n}\to \CC$ is a given self-adjoint kernel defining a trace-class operator $\widehat{\kappa}$, and $\widehat{\kappa}$ has (real) eigenvalues $\{\lambda_i\}_{i\geq 0}$. 
The authors of \cite{Newton} introduced the elementary symmetric polynomials of the eigenvalues, that is, let $e_0=1$ and for each integer $k\geq 1$ let 
\begin{equation} \label{eq:ek} 
e_k \defeq \sum_{0\leq i_1<\dots<i_k} \lambda_{i_1}\cdots \lambda_{i_k}.
\end{equation} 
The following claim might be standard, see \cite[Prop.~2.1]{Newton} for an exact reference.
\begin{claim} The trace-class operator $\widehat{\kappa}$ is positive semidefinite if and only if $e_k\geq 0$ for all $k\geq 1$.
\end{claim}
By the Plemelj--Smithies formulas\footnote{See also Newton's identities or the Faddeev–LeVerrier algorithm.} for a self-adjoint kernel $\kappa\in L^2(\RR^{2n})$ the quantities $e_k$ can be calculated as   
\begin{equation} \label{eq:det}
e_k =\frac{1}{k!}\left|\begin{array}{ccccc}
M_1  & 1   & 0  & \cdots       \\
M_2  & M_1 & 2  & 0  & \cdots \\
\vdots  &  & \ddots & \ddots   \\
M_{k-1} & M_{k-2} & \cdots & M_1    & k-1 \\
M_k     & M_{k-1} & \cdots & M_2    & M_1
\end{array}\right|,
\end{equation}
where the moments $\{M_j\}_{1\leq j\leq k}$ are defined as 
\begin{equation*}
M_j \defeq \sum_{i=0}^{\infty}\lambda_i^j=\Tr\{ \widehat{\kappa}^j\}
=\int_{\mathbb{R}^{jn}} \, \left[\kappa(x_j,x_1) \prod_{i=1}^{j-1} \kappa(x_i,x_{i+1}) \right]\, \mathrm{d}x_1\cdots \mathrm{d}x_j.
\end{equation*}
It was demonstrated in \cite{Newton} that this calculation is effective for polynomial Gaussian operators, and the authors considered the family of kernels 
\begin{equation*}
\kappa_{\gamma}(x,y)=\frac{4}{\sqrt{\pi}(2+\gamma)} \left(\gamma(x+y)^2-(x-y)^2+1\right) \kappa_G(x,y)
\end{equation*}
were considered, where $\gamma\geq 0$ is a parameter\footnote{The positivity of $\widehat{\kappa}_{\gamma}$ requires $\gamma\geq 0$, and $\widehat{\kappa}_{\gamma}$ is easily seen to be positive at $\gamma=1$.}, and
\begin{equation*}
\kappa_G(x,y)=\exp\left[-\frac 32 (x-y)^2-(x+y)^2\right].
\end{equation*}
If $\widehat{\kappa}_{\gamma}\geq 0$ for some $\gamma$, then by Theorem~\ref{t:order} positivity is preserved whenever we replace $\kappa_G$ by another Gaussian kernel $\kappa_{G'}$ such that $\kappa_{G}\preceq \kappa_{G'}$. For the sake of simplicity we only consider Gaussian kernels $\kappa_{G'}$ such that $\kappa_{G'}\approx \kappa_G$. By Theorem~\ref{t:equiv} for all parameters $\gamma\geq 0$ and $\delta>-1$ we obtain the kernels
\begin{equation*}
\kappa^{\delta}_{\gamma}(x,y)=\frac{4}{\sqrt{\pi}(2+\gamma)} \left(\gamma(x+y)^2-(x-y)^2+1\right) \kappa^{\delta}_G(x,y), 
\end{equation*}
where 
\begin{equation*}
\kappa^{\delta}_G(x,y)=\frac{(1+\delta)^{3/2}(2+\gamma)}{2+2\delta+\gamma} \exp\left[-\left(\frac32+\delta\right) (x-y)^2-(1+\delta)(x+y)^2\right]
\end{equation*}
\begin{figure}
\begin{center}
\includegraphics[angle=270,width=\textwidth]{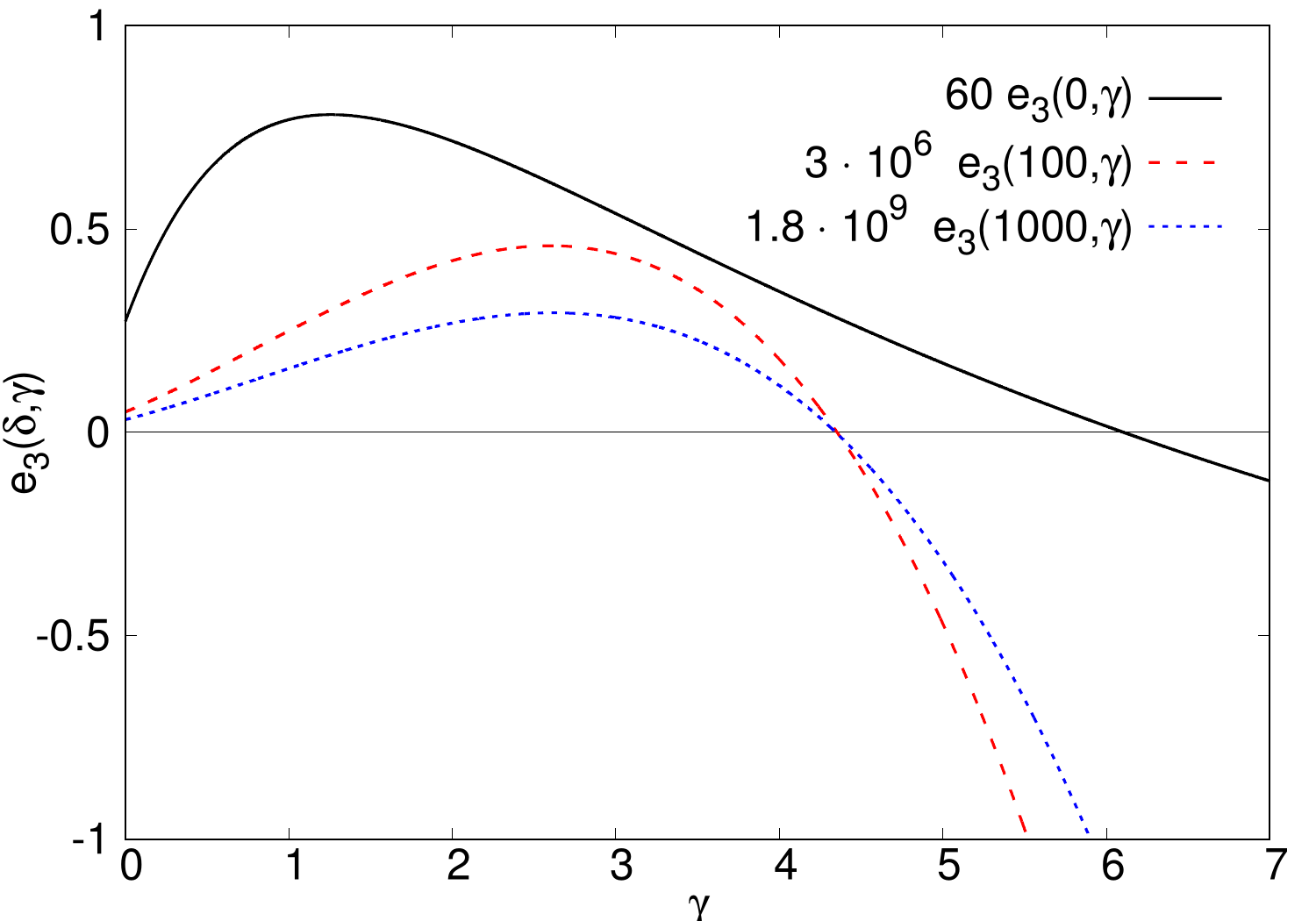}
\end{center}
\caption{The quantities $e_3(\delta,\gamma)$ as a function of $\gamma$ for the parameters $\delta=0,100,1000$ scaled appropriately. As $\delta$ is larger, the root $\iZ_3(\delta)$ becomes smaller, so we obtain a larger set of parameters $\gamma$ where positivity fails.}
\label{fig:ek}
\end{figure}
are our modified Gaussian kernels\footnote{Note that the normalizing constants $\frac{4}{\sqrt{\pi}(2+\gamma)}$ and $\frac{(1+\delta)^{3/2}(2+\gamma)}{2+2\delta+\gamma}$ are only to ensure that $\Tr(\widehat{\kappa}_{\gamma})=\Tr(\widehat{\kappa}^{\delta}_{\gamma})=1$, which was practical for us during the numerical calculations.}. Let us denote by $e_k(\delta,\gamma)$ the values of $e_k$ according to \eqref{eq:ek} for the operator $\widehat{\kappa}^{\delta}_{\gamma}$, see also \eqref{eq:det}. Furthermore, for any integer $k\geq 1$ and each $\delta>-1$ let $\iZ_k(\delta)$ denote the only $\gamma_0\geq 0$ for which $e_k(\delta,\gamma_0)=0$, and let $\iZ_k(\infty)=\displaystyle{\lim_{\delta\to \infty}} \iZ_k(\delta)$\footnote{We only consider parameters $k$ and $\delta$ for which $\iZ_k(\delta)$ is well-defined and the limit $\iZ_k(\infty)$ exists, we conjecture that this is the case in general.}, see Figure~\ref{fig:ek} for the illustration of $e_3(\delta,\gamma)$ for some parameters $\delta$. For our parameters $e_k(\delta,\gamma)$ is positive if $\gamma<\iZ_k(\delta)$ and negative if $\gamma>\iZ_k(\delta)$. Thus $\widehat{\kappa}^{\delta}_{\gamma}$ is not positive semidefinite for $\gamma>\iZ_k(\delta)$, so $\widehat{\kappa}_{\gamma}$ is not positive semidefinite either by the equivalence of the Gaussian kernels $\kappa_G^{\delta}$. This implies that
\begin{equation} \label{eq:Zk}
\text{if $\gamma>\inf\{\iZ_k(\delta): \delta\in (-1,\infty]\}$, then $\widehat{\kappa}_{\gamma}$ is not positive semidefinite}.
\end{equation}
\begin{table} 
\begin{center}
\begin{tabular}{|c | c|c|c|}
\hline 
 $\iZ_k(\delta)$ &  $k=3$  & $k=4$  & $k=5$  \\
\hline
$\delta=0$ &  6.10781  &  5.07931 &  4.25293    \\
\hline
$\delta=10$ & 4.43150  &  4.38966 &  4.05059 \\
\hline
$\delta=50$ & 4.36304  &  4.35449 &  4.04157 \\
\hline
$\delta=250$ & 4.34880 & 4.34708 & 4.03973 \\
\hline
$\delta=\infty$ & 4.34521 & 4.34521 & 4.03924 \\
\hline
\end{tabular} 
\end{center}
\caption{The table contains the values of $\iZ_k(\delta)$ for $k=3,4,5$ and $\delta=0,10,50,250,\infty$. Note that $\iZ_k(\delta)$ is decreasing in $\delta$, and we can see that switching from the original $\iZ_k(0)$ to the optimal $\iZ_k(\infty)$ is a significant improvement for every $k$.}
\label{t:tab}
\end{table}
By \eqref{eq:Zk} our goal is to find the minima of the functions $\delta\mapsto \iZ_k(\delta)$.
We calculated $e_k(\delta,\gamma)$ for $k=3,4,5$ explicitly with the program Mathematica using \eqref{eq:det}; we obtained a positive denominator, and the numerator appeared as a polynomial of degree $k$ in $\gamma$ with the coefficients depending on $\delta$. Then we numerically calculated $\iZ_k(\delta)$ for $k=3,4,5$ by solving the polynomial equations. Surprisingly, we could even calculate $Z_k(\infty)$ without any numerical error by obtaining the limit polynomial as $\delta\to \infty$: we simply divide the polynomial by its leading coefficient and take the limit as $\delta\to \infty$. For $k=3,4$ these limit polynomials are 
\begin{equation*} 
\gamma^3-\frac 52 \gamma^2-\frac{15}{2}\gamma-\frac 94 \quad \text{and} \quad  \gamma^4-\frac{10}{3} \gamma^3-\frac{85}{3}\gamma^2-21\gamma-\frac{15}{4},
\end{equation*} 
respectively. Calculating their only positive roots implies that
\begin{equation*} \iZ_3(\infty)=\iZ_4(\infty)=2+\sqrt{\frac{11}{2}}. 
\end{equation*}
Similarly, for $k=5$ we obtain the limit polynomial 
\begin{equation*} 
P(\gamma)\defeq \gamma^5+\frac{125}{24}\gamma^4-\frac{247}{12} \gamma^3-\frac{245}{4}\gamma^2-\frac{205}{8}\gamma-\frac{75}{32}.
\end{equation*}
Factorizing as
\begin{equation*} P(\gamma)=\frac{1}{96}(6\gamma^2+44\gamma+15)(16\gamma^3-34\gamma^2-120\gamma-15)
\end{equation*} 
yields that 
\begin{equation*}
\iZ_5(\infty)\approx 4.03924 \text{ is the only positive root of } 16\gamma^3-34\gamma^2-120\gamma-15.  
\end{equation*}
See Table~\ref{t:tab} for the summary of these values, and the sizeable gaps between the original $\iZ_k(0)$ and the optimal $\iZ_k(\infty)$. Note that $\iZ_5(\infty)\approx 4.03924$ is close to $\iZ_{8}(0)\approx 4.03021$ calculated in \cite{Newton}. Calculating $e_8(0,\gamma)$ is much more difficult than of $e_5(\delta,\gamma)$ as the length of the formulas grow super-exponentially in $k$, therefore our method turns out to be more efficient than the original one from \cite{Newton}.

\subsection{Entanglement} \label{ss:entanglement}
Here we enclose an introduction of entanglement, then we deduct some consequences of our earlier positivity results. 
A physical system in quantum mechanics can be described by a density operator $\widehat{\rho} \colon \mathcal{H} \rightarrow \mathcal{H}$ acting on a complex Hilbert space $\mathcal{H}$, see e.\,g.~John von Neumann's original works~\cite{Neumann,Neumann1927}. Here we will consider the Hilbert space of square integrable complex functions $\mathcal{H}=L^2\left(\mathbb{R}^{n}\right)$, where $n$ is the number of degrees of freedom in our physical system. Our indices $1,\dots,n$ correspond to variables which might, for example, represent the location or moment of some particles. We can divide our system to $m$ physical subsystems by partitioning our index set $\{1,\dots,n\}$ into $m$ pairwise disjoint sets $P_1,\dots, P_m$ of sizes $d_1,\dots, d_m$, respectively, where $\sum^{m}_{k=1} d_k=n$. For any $k\in \{1,\dots,m\}$ let $\iH_k=L^2\left(\RR^{P_k}\right)\cong L^2(\RR^{d_k})$, which is the Hilbert space belonging to the $k$th physical subsystem. This partition decomposes $\mathcal{H}$ to the tensor product 
\begin{equation*}
\mathcal{H}= \bigotimes_{k=1}^m \mathcal{H}_k.
\end{equation*}
For more information on $C^{*}$-algebras see e.\,g.~\cite{Murphy}.

\begin{definition} \label{d:sep}
Let $\widehat{\rho}\colon L^2\left(\mathbb{R}^{n}\right)\to L^2\left(\mathbb{R}^{n}\right)$ be a density operator and assume that $m$ physical subsystems are given by an $m$-element partition $\mathcal{P}=\{P_1,\dots,P_m\}$ of $\{1,\dots, n\}$, and let $\iH_1,\dots,\iH_m$ be the corresponding Hilbert spaces. We call $\widehat{\rho}$ \emph{separable with respect to $\mathcal{P}$} if it can be written as
\begin{equation}
\widehat{\rho}=\sum_{j=1}^{\infty} \pi_j \bigotimes_{k=1}^m  
  \widehat{\rho}^{\, (k)}_j, \text{ where } \pi_j \geq 0 \text{ and }  \sum_{j=1}^{\infty} \pi_j=1,
\label{eq: separablity_general_case}
\end{equation}
and $\widehat{\rho}^{ \, (k)}_j \colon \iH_k \to \iH_k$ are density operators for all $j\geq 1$. We say that $\widehat{\rho}$ is \emph{entangled with respect to $\mathcal{P}$} if it cannot be written as \eqref{eq: separablity_general_case}. 
\end{definition}

\begin{definition}
Let $\iP$ and $\iP'$ be two partitions of $\{1,\dots,n\}$. We say that $\iP$ is a \emph{refinement} of $\iP'$ if for all $P\in \iP$ there is a $P'\in \iP'$ such that $P\subset P'$. 
\end{definition}

The definition of separability easily implies the following, which allows us to prove multipartite entanglement by proving bipartite entanglement for a coarser two-element partition. 

\begin{fact} Assume that partition $\iP$ is a refinement of $\iP'$. If the density operator $\widehat{\rho}$ is separable with respect to $\iP$, then it is separable with respect to $\iP'$. Equivalently, if $\widehat{\rho}$ is entangled with respect to $\iP'$, then it is entangled with respect to $\iP$, too.
\end{fact}

From now on, we only consider the bipartite case, that is, a $2$-element partition $\iP=\{P_1,P_2\}$ of our index set $\{1,\dots, n\}$ with given sizes $d_1,d_2$. We consider the separability of density operators $\widehat{\rho}$ with respect to $\iP$. By varying the indices if necessary, we may assume without loss of generality that $P_1=\{1,\dots,d_1\}$ and $P_2=\{d_1+1,\dots,n\}$. Recall that the Wigner--Weyl transform \eqref{eq:Wigner-transform} and its inverse \eqref{eq:Inverse Wigner-transform} provide the connection between the phase space and the position representations. If $\widehat{\rho}$ is separable with respect to $\iP$, similarly to \eqref{eq: separablity_general_case} it can be written in phase space representation as
\begin{equation*}
W(x,p)=W\big(x^{(1)},p^{(1)};x^{(2)},p^{(2)}\big) =\sum_{j=1}^{\infty} \pi_j W^{(1)}_j\big(x^{(1)},p^{(1)}\big)  \cdot W^{(2)}_j\big(x^{(2)},p^{(2)}\big),
\end{equation*}
where $\pi_j \geq 0$ and $\sum_{j=1}^{\infty} \pi_j=1$, and $W^{(i)}_j$ is the phase space representation of the density operator $\widehat{\rho}^{\, (i)}_j$ coming from Definition~\ref{d:sep} for all $i\in \{1,2\}$ and $j\geq 1$, and we used the notation 
\begin{align*}
x&=(x_1,\dots, x_n)^{\top},~
x^{(1)}=\left(x_{1},\dots,x_{d_1}\right)^{\top}, \text{ and } x^{(2)}=\left(x_{d_1+1},\dots,x_{n}\right)^{\top} ; \\ 
p&=(p_1,\dots,p_n)^{\top},~
p^{(1)}=\left(p_{1},\dots,p_{d_1}\right)^{\top}, \text{ and } p^{(2)}=\left(p_{d_1+1},\dots,p_{n}\right)^{\top}.
\end{align*} 
The \emph{partial transpose} (PT) in the first partition is a transformation on the Wigner function as follows, see e.\,g.~\cite{Perescrit,SimonR}:
\begin{equation*}
W\big(x^{(1)},p^{(1)};x^{(2)},p^{(2)}\big) \rightarrow \widetilde{W}(x,p)\defeq W\big(x^{(1)},-p^{(1)};x^{(2)},p^{(2)}\big).
\end{equation*}
A similar definition for the partial transpose in the second partition is equivalent for all practical purposes, see e.\,g.~\cite{GUHNE} for more details. In fact, this property easily follows in the position representation. The following necessary condition of separability were independently found by Peres \cite{Perescrit} and the Horodecki family \cite{HORODECKI19961}. 

\begin{theorem}[Peres--Horodecki criterion] \label{t:PH} If $\widehat{\rho}$ is separable, then $\widetilde{W}(x,p)$ defines a density operator.
\end{theorem}

It appears that the another direction of the above theorem also holds if $\widehat{\rho}_G$ is Gaussian and one set in the partition contains a single element, that is, either $d_1=1$ or $d_2=1$. The following theorem was first proved in the case $d_1=d_2=1$ by Simon \cite{SimonR}, and the proof was generalized for the $d_1=1$ or $d_2=1$ case by Werner and Wolf \cite{WernerWolf}. Later Lami, Serafini and Adesso \cite{Lami} simplified the proof and summarized the current status of knowledge about this topic.

\begin{theorem}[Werner--Wolf] \label{t:WW} Let $\widehat{\rho}_G$ be Gaussian and let $d_1=1$ or $d_2=1$. Then $\widehat{\rho}_G$ is separable if and only if $\widetilde{W}_G(x,p)$ represents a Gaussian density operator.
\end{theorem}
 
The kernel $\rho(x,y)$ in position representation can be given according to \eqref{eq:Inverse Wigner-transform}, which is separable if and only if we can write it in the following form:
\begin{equation*}
\rho(x,y): =\rho\big(x^{(1)},y^{(1)};x^{(2)},y^{(2)}\big)=\sum_{j=1}^{\infty} \pi_j \rho^{(1)}_j\big(x^{(1)},y^{(1)}\big) \rho^{(2)}_j\big(x^{(2)},y^{(2)}\big),
\end{equation*}
where $\pi_j \geq 0$ and $\sum_{j=1}^{\infty} \pi_j=1$, and $\rho^{(i)}_j$ is the position representation of the density operator $\widehat{\rho}^{\, (i)}_j$ for all $i\in \{1,2\}$ and $j\geq 1$, and we used the notation 
\begin{align*}
x&=(x_1,\dots, x_n)^{\top},~
x^{(1)}=\left(x_{1},\dots,x_{d_1}\right)^{\top}, \text{ and } x^{(2)}=\left(x_{d_1+1},\dots,x_{n}\right)^{\top} ; \\ 
y&=(y_1,\dots,y_n)^{\top},~
y^{(1)}=\left(y_{1},\dots,y_{d_1}\right)^{\top}, \text{ and } y^{(2)}=\left(y_{d_1+1},\dots,y_{n}\right)^{\top},
\end{align*} 
see e.\,g.~\cite[Section 6.3]{Murphy}. 
The condition $\Tr(\widehat{\rho})=1$ in position representation means
\begin{equation} \label{eq:norm}
\int_{\mathbb{R}^n} \rho(x,x) \, \mathrm{d} x=1.
\end{equation}
The PT operation in position representation is given as
\begin{equation*}
\rho(x,y)=\rho\big(x^{(1)},y^{(1)};x^{(2)},y^{(2)}) \rightarrow 
 \widetilde{\rho}(x,y)\defeq \rho\big(y^{(1)},x^{(1)};x^{(2)},y^{(2)}\big).
 \end{equation*}
Note that if ${\widetilde{\rho}}_2$ denotes the PT operation in the second coordinates then
\begin{equation*}
\widetilde{\rho}(x,y)=\rho\big(y^{(1)},x^{(1)};x^{(2)},y^{(2)}\big)=\rho^{*}\big(x^{(1)},y^{(1)};y^{(2)},x^{(2)}\big)=({\widetilde{\rho}}_2(x,y))^{*},
\end{equation*}
and conjugating the kernel does not change the spectrum of a self-adjoint operator, so the two possible forms of the PT transform are equivalent for our investigation. For the sake of completeness, here we enclose Theorems~\ref{t:PH} and \ref{t:WW} in position representation, too.

\begin{definition}
A density operator $\widehat{\rho}$ is called \emph{PPT} (\emph{positive partial transpose}), if its partial transpose $\widetilde{\rho}(x,y)$ is a kernel of a density operator. A density operator is called \emph{NPT} (\emph{non-positive partial transpose}) if it is not \emph{PPT}.
\end{definition}

\begin{theorem}[Peres--Horodecki criterion] If $\widehat{\rho}$ is a separable, then it is PPT.
\end{theorem}

\begin{theorem}[Werner--Wolf] \label{t:WW2} Let $\widehat{\rho}_G$ be Gaussian and let $d_1=1$ or $d_2=1$. Then $\widehat{\rho}_G$ is separable if and only if it is PPT.
\label{Th: Werner}
\end{theorem}

\begin{definition} Let $P$ be a self-adjoint polynomial and $\widehat{\rho}_G$ be a Gaussian density operator and assume that $P\rho_G$ is a kernel of a positive operator. Then we can define the \emph{polynomial Gaussian density operator} $\widehat{\rho}_{\PG}$ such that its kernel $\rho_{\PG}$ is a constant multiple of $P\rho_{G}$, see \eqref{eq:norm} for the normalizing factor.
\end{definition}

Recall the preorder $\preceq$ from Definition~\ref{d:prec}. The next corollary easily follows from Theorem~\ref{t:order}. It means that if we have an NPT polynomial Gaussian operator, then we can find a lot of others by changing only its Gaussian factor. 

\begin{corollary} \label{c:order}
Let $\widehat{\rho}_{G_0}$, $\widehat{\rho}_{G_1}$ be Gaussian density operators with $\widetilde{\rho}_{G_0}\preceq \widetilde{\rho}_{G_1}$. Let $P$ be a self-adjoint polynomial and assume that $\widehat{\rho}_{\PG_0}$ and $\widehat{\rho}_{\PG_1}$ exist. Then
\begin{equation*}
\widehat{\rho}_{\PG_1} \text{ is NPT} \quad  \Longrightarrow  \quad \widehat{\rho}_{\PG_0} \text{ is NPT}. 
\end{equation*}
\end{corollary}

From now on let $P$ be a self-adjoint polynomial and let $\widehat{\rho}_G$ be a Gaussian density operator such that the corresponding polynomial Gaussian density operator $\widehat{\rho}_{\PG}$ exists. Theorem~\ref{Th:BHCS} yields the next corollary. 

\begin{corollary} \label{c:NPT}
If $\widehat{\rho}_G$ is NPT, then $\widehat{\rho}_{\PG}$ is NPT as well.
\end{corollary}

Theorems~\ref{Th:BHCS} and \ref{Th: Werner} imply the following corollary. 
 
\begin{corollary}  \label{c:NPT2}
If $d_1=1$ or $d_2=1$ and $\widehat{\rho}_{G}$ is entangled, then $\widehat{\rho}_{\PG}$ is entangled. 
\end{corollary} 

The following problem asks whether Theorem~\ref{t:WW2} can be generalized to polynomial Gaussian operators, that is, in the $d_1=1$ or $d_2=1$ case the entanglement reduces to a positivity check. 

 \begin{problem} \label{p:problem}Let $d_1=1$ or $d_2=1$. If $\widehat{\rho}_{\PG}$ is entangled, is it necessarily NPT? In other words, is it true that $\widehat{\rho}_{\PG}$ is separable if and only if it is PPT?
\end{problem}

\section{When the Gaussian is not positive} \label{s:Gauss}

The main goal of this section is to prove Theorem~\ref{Th:BHCS}. First we will settle the one-dimensional case in Subsection~\ref{ss:one}. Then in Subsection~\ref{ss:W} we recall the symplectic decomposition for Gaussian kernels, in particular we state Claim~\ref{c:positivity} which will be useful in Section~\ref{s:order}, too. Finally, in Subsection~\ref{ss:gen} we trace back the general case to the one-dimensional result using a symplectic decomposition based on Williamson's theorem \cite{Williamson_origin}.

\subsection{The one-dimensional case} \label{ss:one}

In this subsection we prove the following. 

\begin{theorem} \label{theorem:non_positive_gaussian} Let $P\colon \RR^2\to \CC$ be a non-zero polynomial and let $\kappa_G\colon \RR^2\to \RR$ be a Gaussian kernel of the form 
\begin{equation*}\kappa_G(x,y)=\exp\left(-A(x-y)^2-C(x+y)^2\right)
\end{equation*} 
such that $\widehat{\kappa}_G$ is not positive semidefinite. Then $\widehat{\kappa}_{\PG}$ is not positive semidefinite. 
\end{theorem}

\begin{proof} 
Assume to the contrary that $\widehat{\kappa}_G$ is not positive semidefinite and $\widehat{\kappa}_{\PG}$ is positive semidefinite. It is well known that $\widehat{\kappa}_G$ is positive semidefinite if and only if $A\geq C>0$, see e.\,g.~\cite{Newton} where even the eigenfunctions are calculated. Thus we obtain that $C>A>0$.

First we show that $R(x)=P(x,x)$ is a non-zero polynomial. Indeed, assume to the contrary that $R(x)=P(x,x)=0$ for all $x\in \RR$, then by \eqref{eq:trace} the sum of the eigenvalues equals to 
\begin{equation} \label{eq:tr} \Tr \left(\widehat{\kappa}_{\PG}\right)=\int_{-\infty}^{\infty} \kappa_{\PG}(x,x) \,\mathrm{d} x=0. 
\end{equation} 
As $\widehat{\kappa}_{\PG}$ is not the zero operator, it has non-zero eigenvalues. Hence $\widehat{\kappa}_{\PG}$ has a negative eigenvalue by \eqref{eq:tr}, which contradicts that it is positive semidefinite. 
	
Now define the polynomial $Q\colon \RR^2\to \CC$ as \begin{equation*} 
Q(x,y)=P(x,-x+y)+P(-x+y,x).
\end{equation*} 
As $Q(x,2x)=2P(x,x)$, it follows that $Q$ is not the zero polynomial. Therefore, we can fix $\varepsilon\in \RR$ such that $Q(x,\varepsilon)$ is a non-zero polynomial of $x$. Applying Theorem~\ref{t:Mercer} for $\kappa_{\PG}$ with $n=1$, $k=2$, $c_1=1$, and $c_2=-1$ implies that for all $x_1,x_2\in \RR$ we have
\begin{equation} \label{eq:kappa} 
\kappa_{\PG}(x_1,x_1)+\kappa_{\PG}(x_2,x_2)\geq 
\kappa_{\PG}(x_1,x_2)+\kappa_{\PG}(x_2,x_1).
\end{equation}
Substituting $x_1=x$ and $x_2=\varepsilon-x$ into \eqref{eq:kappa} and using the notation 
\begin{equation} \label{eq:E(x)} 
E(x)=\exp\left(4(A-C)x^2-4A\varepsilon x+(A+C)\varepsilon^2\right)
\end{equation} 
we obtain that for all $x\in \RR$ we have
\begin{equation} \label{eq:AC} 
E(x)\left[R(x)+R(\varepsilon-x)\exp(8C\varepsilon x-4C\varepsilon^2)\right]
\geq Q(x,\varepsilon).
\end{equation} 
The term $4(A-C)x^2$ in \eqref{eq:E(x)} and $A-C<0$ imply that the left hand side of \eqref{eq:AC} tends to $0$ as $x\to \infty$. As $Q(x,\varepsilon)$ is a non-zero polynomial of $x$, the limit $\displaystyle {\lim_{x\to \infty} Q(x,\varepsilon)\neq 0}$ exists, because $Q$ is either constant or the limit is $+\infty$ or $-\infty$. These and \eqref{eq:AC} imply that 
\begin{equation} \label{eq:Qneg} 
\lim_{x\to \infty} Q(x,\varepsilon)<0. 
\end{equation} 
Now consider the kernel 
\begin{equation*} \tau(x,y)=xy\kappa_{\PG}(x,y).
\end{equation*} 
As $\kappa_{\PG}$ satisfies the Mercer-type inequalities \eqref{eq:Mercer}, so does $\tau$. Therefore, similarly to \eqref{eq:kappa} for all $x_1,x_2\in \RR$ we obtain that
\begin{equation} \label{eq:tau} 
\tau(x_1,x_1)+\tau(x_2,x_2)\geq 
\tau(x_1,x_2)+\tau(x_2,x_1).
\end{equation}
Substituting $x_1=x$ and $x_2=\varepsilon-x$ into \eqref{eq:tau} yields that for all $x\in \RR$ we have
\begin{equation} \label{eq:AC2} 
E(x)\left[ x^2R(x)+(\varepsilon-x)^2R(\varepsilon-x) \exp(8C\varepsilon x-4C\varepsilon^2) \right]\geq x(\varepsilon-x)Q(x,\varepsilon).
\end{equation} 
The term $4(A-C)x^2$ in \eqref{eq:E(x)} and $A-C<0$ imply that the left hand side of \eqref{eq:AC2} tends to $0$ as $x\to \infty$. Since $\displaystyle {\lim_{x\to \infty} x(\varepsilon-x)= -\infty}$, inequality \eqref{eq:Qneg} yields $\displaystyle {\lim_{x\to \infty} x(\varepsilon-x)Q(x,\varepsilon)=+\infty}$, which clearly contradicts \eqref{eq:AC2}.
\end{proof} 

\begin{remark} The above proof only yields in higher dimensions that if $\widehat{\kappa}_{\PG}$ is positive semidefinite, then $A\geq C$ holds for the Gaussian parameters. This is not sufficient for the positivity of the Gaussian $\widehat{\kappa}_{G}$ if the matrix $B$ is not symmetric, see e.\,g.~Fact~\ref{f:IB}. In order to describe the positivity of a multidimensional Gaussian, we need to consider the symplectic eigenvalues as stated in Claim~\ref{c:positivity}, too. Therefore, symplectic methods come naturally into the picture, and they will be used in the next two subsections to extend Theorem~\ref{theorem:non_positive_gaussian} to higher dimensions. 
\end{remark}

\subsection{Symplectic decomposition and Williamson's theorem} \label{ss:W}
The results of this subsection are known, but we will need them in Subsection~\ref{ss:gen} and also in Section~\ref{s:order}.  
Recall that our Gaussian kernel $\kappa_{G}\colon \RR^{2n}\to \CC$ is of the form
\begin{equation*} 
\kappa_{G}(x,y)=\exp\left\{-(x-y)^{\top} A (x-y) -i(x-y)^{\top} B (x+y)-(x+y)^{\top} C (x+y)\right\},
\end{equation*}
where $A,B,C\in \RR^{n\times n}$ such that $A$ and $C$ are symmetric and positive definite; if $n=1$ then $\widehat{\kappa}_G$ is positive semidefinite if and only if $A\geq C>0$.

Calculating the Wigner--Weyl transform of $\kappa_G$ (see \cite[(2.18) and (2.19)]{SimonWW} with a slightly different terminology) yields the following formula in the phase space:
\begin{equation} \label{eq:WG}
W_G(x,p)=c_{G} \exp\left\lbrace -v^{\top}Gv \right\rbrace, 
\end{equation}
where $c_G=2^n\pi^{-\frac 32 n}(\det C)^{\frac 32} (\det A)^{-\frac{1}{2}}$, and $v=(x,p)^{\top}$, and 
\begin{equation} \label{eq:G} 
G=\begin{pmatrix}
4 C + B^{\top} A^{-1} B && \frac{1}{2} B^{\top}A^{-1} \\
\frac{1}{2} A^{-1} B  && \frac{1}{4} A^{-1}
\end{pmatrix}  
\end{equation}
is a symmetric, positive definite\footnote{For the positivity of $G$ one can simply check that $v^{\top}Gv\geq 0$ for all $v\in \RR^{2n}$; it also follows from the fact that the Wigner--Weyl transform maps $L^2(\RR^{2n})$ to $L^2(\RR^{2n})$, so $W$ is square integrable.} real $2n \times 2n$ matrix\footnote{In the literature $G^{-1}$ is called the \emph{Gaussian covariance matrix}.}. Define the $2n\times2n$ matrix
\begin{equation*} 
\Omega=\begin{pmatrix}
0 & I_n \\
-I_n  & 0
\end{pmatrix},  
\end{equation*}
where $I_n$ is the $n \times n$ identity matrix and $0$ is the $n\times n$ zero matrix. A $2n \times 2n$ real matrix $S\in \Sp(2n,\RR)$ is called \emph{symplectic} if it satisfies $S^{\top} \Omega S = \Omega$.

For the following important theorem see \cite[Theorem~215]{Gosson_Harmonic}, which states that applying a symplectic transform in  the phase space does not change the spectrum of the operator. 
\begin{theorem}[Weyl correspondence] \label{t:Weyl}
Assume that $\kappa\in L^2(\RR^{2n})$ and $W$ is its Wigner--Weyl transform, and let $S\in \Sp(2n,\RR)$ be a symplectic matrix. Then the operator corresponding to $W\circ S$ equals $\widehat{S}^{-1}  \widehat{\kappa}  \widehat{S}$ with some operator\footnote{Note that the operator $\widehat{S}$ is explicitly constructed in \cite{deGosson_book,Serafini,_pereira}.} $\widehat{S}$; note that $\widehat{\kappa}$ and $\widehat{S}^{-1}  \widehat{\kappa}  \widehat{S}$ are in the same conjugacy class.   \end{theorem}

The following theorem is due to Williamson \cite{Williamson_origin}, see also \cite[Subsection~3.2.3]{Serafini}, \cite[Theorem~93]{Gosson_Harmonic}, or \cite[Theorem~8.11]{deGosson_book}.

\begin{theorem}[Williamson's theorem] \label{t:Williamson}
Let $G\in \RR^{2n\times 2n}$ be a symmetric, positive definite matrix. There is a $2n\times 2n$ symplectic matrix\footnote{Note that finding such a matrix $S$ is not easy in practice, see \cite{_pereira} for tackling this problem.} $S\in \Sp(2n,\RR)$ such that
\begin{equation} \label{eq:diag}
S^{\top}GS=\begin{pmatrix}
  \Lambda & 0 \\ 
  0 & \Lambda
\end{pmatrix}, 
\end{equation}
where $\Lambda$ is the $n\times n$ diagonal matrix
\begin{equation*} \Lambda =  \begin{bmatrix}\mu_{1} & & \\ & \ddots & \\   & & \mu_{n}\end{bmatrix}. \end{equation*}
Furthermore, $\mu_k$ are positive and $\pm \mu_k$ are the eigenvalues\footnote{The reciprocals $1/\mu_i$ are called the \emph{symplectic eigenvalues} of the Gaussian covariance matrix.} of the matrix $M=iG\Omega$. 
\end{theorem}

By Theorem~\ref{t:Williamson} we can diagonalize $G$ from \eqref{eq:G} with a symplectic matrix, that is, there is a symplectic matrix $S\in \Sp(2n,\RR)$ such that \eqref{eq:diag} holds. We apply the linear transformation $(x',p')^{\top}=v'=Sv$ in the phase space; by Theorem~\ref{t:Weyl} this does not change the spectrum of $\widehat{\kappa}_{G}$. Now \eqref{eq:WG} and \eqref{eq:diag} imply that the linear transformation $v\mapsto Sv$ yields the product formula 
\begin{equation*}
W_G(x',p')=c_G \prod^n_{k=1} W_{G_k}(x'_k,p'_k),
\end{equation*}
where 
\begin{equation} \label{eq:wgk} W_{G_k}(x,p)=\exp\left\{-\mu_k(x^2+p^2)\right\}.
\end{equation} 
The inverse Wigner transformation \eqref{eq:Inverse Wigner-transform} allows us to  factorize our kernel in position representation as 
\begin{equation}  \label{eq:K'G}
\kappa'_G(x',y')=c_G\prod^n_{k=1} \kappa_{G_k}(x'_k,y'_k),
\end{equation}
where \eqref{eq:wgk} yields
\begin{equation} \label{eq:kgk} 
\kappa_{G_k}(x,y)=
\left(\frac{\pi}{\mu_k}\right)^{\frac 12} \exp\left\{-\frac{1}{4\mu_k} (x-y)^2-\frac{\mu_k}{4}(x+y)^2\right\}.
\end{equation} 
Clearly, $\widehat{\kappa}'_G\geq 0$ if and only if $\widehat{\kappa}_{G_k}\geq 0$ for all $k$, and by \eqref{eq:kgk} this is equivalent to $\mu_k\leq 1$ for all $1\leq k\leq n$. We emphasize this result in the following claim. 
\begin{claim}[Positivity of Gaussian operators] \label{c:positivity}
The Gaussian operator $\widehat{\kappa}_G$ is positive semidefinite if and only if all the eigenvalues $\mu_k$ from Theorem~\ref{t:Williamson} satisfy $\mu_k\leq 1$. 
\end{claim}

\subsection{The general case} \label{ss:gen}

Before proving Theorem~\ref{Th:BHCS} we recall the following notion tailor-made for our need. Define the \emph{partial trace}\footnote{In physics the partial trace operation is applied for kernels of density operators.} (see \cite{Partial_trace}) of a kernel $\kappa\in \iS(\RR^{2n})$ in the coordinates $x_2,\dots,x_n$ as the kernel $\eta\in \iS(\RR^2)$ satisfying
\begin{equation*}
\eta(x_1,y_1)=\int_{\RR^{n-1}} \kappa(x_1,x_2,\dots,x_n, y_1, x_2,\dots,x_n) \, \mathrm{d} x_2\cdots \, \mathrm{d} x_n.
\end{equation*} 
\begin{fact} \label{f:eta}
$\widehat{\kappa}$ and $\widehat{\eta}$ are trace-class with $\Tr(\widehat{\eta})=\Tr(\widehat{\kappa})$, and $\widehat{\kappa}\geq 0$ implies $\widehat{\eta}\geq 0$.
\end{fact}

\begin{proof} As $\kappa\in \iS(\RR^{2n})$ implies $\eta\in \iS(\RR^2)$, we obtain that $\widehat{\kappa}$ and $\widehat{\eta}$ are trace-class. Applying \eqref{eq:trace} for both $\widehat{\eta}$ and $\widehat{\kappa}$, and using the definition of $\eta$ imply
\begin{equation*}
\Tr(\widehat{\eta})=\int_{\RR} \eta(x_1,x_1) \, \mathrm{d}x_1=\int_{\RR^n} \kappa(x,x) \, \mathrm{d}x=\Tr(\widehat{\kappa}).
\end{equation*}
Now assume that $\widehat{\kappa}\geq 0$. Let $x_1,\dots,x_k\in \RR$ and $c_1,\dots,c_k\in \CC$ be arbitrary. Theorem~\ref{t:Mercer} implies that for every $u\in \RR^{n-1}$ we have
 \begin{equation} \label{eq:kappapos}
 \sum_{i,j=1}^k c_i c^{*}_j \kappa(x_i,u;x_j,u)\geq 0. 
 \end{equation} 
The definition of $\eta$ and integrating both sides of \eqref{eq:kappapos} with respect to $u$ imply
\begin{equation*}
\sum_{i,j=1}^k c_i c^{*}_j \eta(x_i,x_j)=\sum_{i,j=1}^k c_i c^{*}_j \int_{\RR^{n-1}} 
 \kappa(x_i,u;x_j,u) \, \mathrm{d}u\geq 0.
\end{equation*}
As $x_1,\dots,x_k\in \RR$ and $c_1,\dots,c_k\in \CC$ were arbitrary, Theorem~\ref{t:Mercer} yields $\widehat{\eta}\geq 0$.
\end{proof}

Now we are ready to prove Theorem~\ref{Th:BHCS}.

\begin{BHCS} 
Let $P\colon \RR^{2n}\to \CC$ be a non-zero polynomial and let $\kappa_G\colon \RR^{2n}\to \CC$ be a Gaussian kernel. If $\widehat{\kappa}_G$ is not positive semidefinite, then $\widehat{\kappa}_{\PG}$ is not positive semidefinite, either. 
\end{BHCS}

\begin{proof}
Assume to the contrary that $\widehat{\kappa}_G$ is not positive semidefinite but $\widehat{\kappa}_{\PG}$ is positive semidefinite. Recall that $\kappa_G$ in coordinate representation is 
\begin{equation*}
\kappa_G(x,y)= 
 \exp\left\{-(x-y)^{\top} A (x-y) -i(x-y)^{\top} B (x+y)-(x+y)^{\top} C (x+y)\right\},
\end{equation*}
where $A,B,C\in \RR^{n\times n}$, and $A,C$ are symmetric and positive definite. We will use the notation and results of Subsection~\ref{ss:W}. We have seen in \eqref{eq:WG} that the Wigner--Weyl transform of $\kappa_G$ is 
\begin{equation*}
W_G(x,p)=c_{G} \exp \left\{ -v^{\top}Gv \right\}, 
\end{equation*}
where $v=(x,p)^{\top}$, furthermore the constant $c_G>0$, and the symmetric, positive definite matrix $G\in \RR^{2n\times 2n}$ are also given there. 
By differentiation under the integral sign we obtain that the Wigner--Weyl transform of $\kappa_{\PG}$ satisfies
\begin{equation} \label{eq:WPG}
W_{\PG}(x,p)=Q(x,p)\exp \left\{ -v^{\top}Gv \right\}, 
\end{equation}
where $Q$ is a polynomial in $2n$ variables. By Theorem~\ref{t:Williamson} we can diagonalize $G$ by a symplectic matrix $S\in \Sp(2n,\RR)$, that is, 
\begin{equation} \label{eq:SGS}
S^{\top}GS=\begin{pmatrix}
  \Lambda & 0 \\ 
  0 & \Lambda
\end{pmatrix},
\end{equation} 
where $\Lambda\in \RR^{n\times n}$ is a diagonal matrix with positive diagonal entries $\mu_k$. 

We will apply the linear transformation $(x',p')^{\top}=v'=Sv$ in the phase space. By Theorem~\ref{t:Weyl} the operator corresponding to $W_{\PG}\circ S$ has the same spectrum as $\widehat{\kappa}_{\PG}$, so it remains positive semidefinite. Now \eqref{eq:WPG} and \eqref{eq:SGS} imply that the linear transformation $v\mapsto Sv$ yields the product formula 
\begin{equation} \label{eq:Wprod}
W_{\PG}(x',p')=Q(x,p)\prod^n_{k=1} W_{G_k}(x'_k,p'_k),
\end{equation}
where 
\begin{equation*} W_{G_k}(x,p)=\exp\left\{-\mu_k(x^2+p^2)\right\}.
\end{equation*} 
The inverse Wigner--Weyl transform \eqref{eq:Inverse Wigner-transform} of the Gaussian part $\prod_{k=1}^n W_{G_k}$ of \eqref{eq:Wprod} was already calculated in \eqref{eq:WG}. By differentiating under the integral sign we easily obtain that inverse Wigner--Weyl transform of \eqref{eq:Wprod} only changes by a polynomial factor, that is, there is a polynomial $R$ in $2n$ variables such that our kernel in position representation is 
\begin{equation*}
\kappa'_{\PG}(x',y')=R(x',y')
\prod^n_{k=1} \kappa_{G_k}(x'_k,y'_k),
\end{equation*}
where $\kappa_{G_k}$ are the one-dimensional Gaussian kernels calculated in \eqref{eq:kgk}. As $\widehat{\kappa}_G$ is not positive semidefinite, and $\kappa'_{G}$ is the product of the same factors $\kappa_{G_k}$ by \eqref{eq:K'G}, we obtain that there exists $k$ such that $\widehat{\kappa}_{G_k}$ is not positive semidefinite. We may assume without loss of generality that this $k$ equals $1$. Now we can define the two-variable polynomial 
\begin{equation*}
S(x'_1,y'_1)=\int_{\RR^{n-1}}  R(x'_1,\dots,x'_n,y'_1,x'_2,\dots,x_n') \prod^n_{k=2} \kappa_{G_k}(x'_k,x'_k) \, \mathrm{d} x'_2\cdots \, \mathrm{d} x'_n.
\end{equation*} 
Consider the partial trace of $\kappa'_{\PG}$ in the coordinates $x'_2,\dots,x'_n$, which yields the one-dimensional kernel $\eta$ given by 
\begin{equation*}
\eta(x_1',y_1')=S(x_1',y_1') \kappa_{G_1}(x_1',y_1'). 
\end{equation*}
By Fact~\ref{f:eta} the partial trace operator preserves positivity and the trace as well, so $\widehat{\eta}\geq 0$ and $S$ is not the zero polynomial. However, $\eta$ is the product of $S$ and a one-dimensional Gaussian kernel $\kappa_{G_1}$ such that $\widehat{\kappa}_{G_1}$ is not positive semidefinite, so the positivity of $\widehat{\eta}$ contradicts the one-dimensional result Theorem~\ref{theorem:non_positive_gaussian}. 
 \end{proof}

\section{When the polynomial is of odd degree} \label{s:odd}

The main goal of this section is to prove Theorem~\ref{t:odd}, but first we need some preparation. We will need the Fourier transformation $\iF\colon L^2(\RR^{2n})\to L^2(\RR^{2n})$, here we use the convention that for all $f\in L^2(\RR^{2n})$ and $u,v\in \RR^{n}$ we have
\begin{equation*}  
\iF(f)(u,v)=\iint_{\RR^{2n}} f(x,y)e^{-i\left(u^{\top}x+v^{\top}y\right)} \, \mathrm{d} x \, \mathrm{d} y.
\end{equation*}

We define the integral transformation $\iL$ for polynomial Gaussian kernels $\kappa_{\PG}$ such that for all $u,v\in \CC^n$ we have
\begin{equation} \label{eq:PG} 
\iL(\kappa_{\PG})(u,v)=\iint_{\RR^{2n}}  \kappa_{\PG}(x,y) \exp\left(-iu^{\top}(x-y)-v^{\top}(x+y)\right) \, \mathrm{d} x 
\, \mathrm{d} y. 
\end{equation}

In the following lemma we calculate the $\iL$-transform of a Gaussian kernel. 

\begin{lemma} \label{l:uv} There is a positive constant $c=c(A,B,C)$ such that for all $u,v\in \CC^n$ our integral $\iL(\kappa_G)(u,v)$ equals to
\begin{equation*} 
c\cdot \exp \left[-\left(\frac{1}{2}BC^{-1}v-u\right)^{\top} \left(4A+ BC^{-1}B^{\top}\right)^{-1} \left(\frac{1}{2}BC^{-1}v-u\right)+\frac 14 v^{\top}C^{-1}v \right].
\end{equation*}
\end{lemma}
\begin{proof} 
For a symmetric, positive definite real matrix $M$ and $b\in \CC^n$ the following formula is well known, see e.\,g.~\cite[Section~I.2.~Appendix 2]{Zee2010} for the proof: 
\begin{equation} \label{eq: gaussian_integral}
\int_{\RR^n} \exp \left[-x^{\top}Mx+b^{\top}x \right]\, \mathrm{d} x = \frac{\pi^{\frac{n}{2}}}{\sqrt{\det M}}\exp{\left(\frac{1}{4}b^{\top} M^{-1}b \right)}.
\end{equation}
First we make a change of variables from $x,y$ to $r=x-y$ and $R=\frac{x+y}{2}$, note that the determinant of its Jacobian matrix has absolute value $1$. Then using \eqref{eq: gaussian_integral} with $M_1=4C$ and $b_1=-2(v+iB^{\top}r)$, and again with $M_2=A+\frac 14 BC^{-1}B^{\top}$ and $b_2=i\left(\frac 12 BC^{-1}v-u\right)$ we obtain that
\begin{align*}
&\iL(\kappa_G)(u,v)=\iint_{\RR^{2n}} \exp \left[-r^{\top} A r -2ir^{\top} BR-4R^{\top} CR-iu^{\top}r-2v^{\top}R\right] \, \mathrm{d} R 
\, \mathrm{d} r \\
&=\frac{\pi^{\frac n2}\exp\left(\frac 14 v^{\top}C^{-1}v\right)}{\sqrt{\det(4C)}} 
\int_{\RR^n} \exp\left[-r^{\top}\left(A+\frac 14 BC^{-1}B^{\top}\right)r+i\left(\frac 12 v^{\top}C^{-1}v\right)r\right] \, \mathrm{d} r \\
&=c\exp \left[-\left(\frac{1}{2}BC^{-1}v-u\right)^{\top} \left(4A+ BC^{-1}B^{\top}\right)^{-1} \left(\frac{1}{2}BC^{-1}v-u\right)+\frac 14 v^{\top}C^{-1}v \right],
\end{align*}
where $c=\pi^{n}\left(\det\left(4AC+BC^{-1}B^{\top}C\right)\right)^{-\frac 12}$. This finishes the proof. 
\end{proof}

Before proving Theorem~\ref{t:odd} we also need a couple of facts. 

\begin{fact} \label{f:Q}
Let $Q\colon \RR^{2n}\to \CC$ be a polynomial with $2n$ real variables.
\begin{enumerate}[(i)] 
\item \label{i:real} If $Q$ takes only real values, then all of its coefficients are real;
\item \label{i:neg} if $Q$ has real coefficients and $\deg Q$ is odd, then $Q$ takes negative values.
\end{enumerate}
\end{fact}

\begin{proof}
First we prove \eqref{i:real}. We can write $Q=Q_1+iQ_2$, where the polynomials $Q_1,Q_2\colon \RR^{2n}\to \RR$ have real coefficients. As $Q$ takes only real values, we obtain that $Q_2(x,y)=0$ for all $x,y\in \RR^n$, which implies that $Q_2\equiv 0$. Then $Q=Q_1$, and the claim follows. 

Now we prove \eqref{i:neg}. Let $d=\deg Q$, which is odd by our assumption. Let us decompose $Q$ as $Q=Q_1+Q_2$ such that $Q_1,Q_2\colon \RR^{2n}\to \RR$ are polynomials with real coefficients, and $Q_1$ contains all the monomials of $Q$ with degree $d$. Then clearly $\deg Q_2<d$ and $Q_1$ is homogeneous of degree $d$, that is, $Q_1(cx,cy)=c^dQ_1(x,y)$ for all $c\in \RR$ and $x,y\in \RR^n$. Fix $x_0,y_0\in \RR^n$ such that $Q_1(x_0,y_0)\neq 0$, and let $r_0=Q_1(x_0,y_0)$. Let us define the one-variable polynomial $R\colon \RR\to \RR$ as 
\begin{equation*} R(c)=Q(cx_0,cy_0)=r_0c^d+Q_2(cx_0,cy_0).
\end{equation*} 
As $\deg Q_2<d$, we have $\deg R=d$ and its leading coefficient is $r_0\neq 0$. As $d$ is odd, $R$ takes negative values, which implies that $Q$ takes negative values as well. 
\end{proof}

\begin{fact} \label{f:L1} 
Assume that $\kappa\in L^2(\RR^{2n})$ is a self-adjoint kernel such that $\widehat{\kappa}$ is positive semidefinite. Let $g\colon \RR^n \to \CC$ be a continuous function and suppose that the function $\tau(x,y)=g(x)g^{*}(y)\kappa(x,y)$ satisfies $\tau\in L^1(\RR^{2n})$.  Then $\iint_{\RR^{2n}} \tau(x,y) \, \mathrm{d} x 	\, \mathrm{d} y \geq 0$.  
\end{fact}

\begin{proof}
As $\tau\in L^1(\RR^{2n})$, the integral $\iint_{\RR^{2n}} \tau(x,y) \, \mathrm{d} x 
\, \mathrm{d} y$ exists. For positive integers $n$ define $g_n\colon \RR^{n} \to \CC$ such that $g_n(x)=g(x)$ if $|x|\leq n$ and $g_n(x)=0$ otherwise. As $g$ is continuous, we have $g_n\in L^2(\RR^n)$. Therefore, the positivity of $\widehat{\kappa}$ implies that for all $n$ we have 
\begin{equation*}
0\leq \iint_{\RR^{2n}} g_n(x)g_n^{*}(y) \kappa(x,y)   \, \mathrm{d} x 	\, \mathrm{d} y \to \iint_{\RR^{2n}} \tau(x,y) \, \mathrm{d} x 	\, \mathrm{d} y
\end{equation*} 
as $n\to \infty$, hence $\iint_{\RR^{2n}} \tau(x,y) \, \mathrm{d} x 	\, \mathrm{d} y \geq 0$ holds.
\end{proof}
Now we are able to prove Theorem~\ref{t:odd}. 

\begin{odd} Assume that $P\colon \RR^{2n}\to \CC$ is a polynomial of odd degree, and $\kappa_G\colon \RR^{2n}\to \CC$ is any Gaussian kernel. Then $\widehat{\kappa}_{\PG}$ is not positive semidefinite. 
\end{odd}
\begin{proof}
Assume to the contrary that $\widehat{\kappa}_{\PG}$ is positive semidefinite. Recall the integral transformation $\iL$ from \eqref{eq:PG}, and note that \begin{equation*} 
\exp\left(-iu^{\top}(x-y)-v^{\top}(x+y)\right)=g_{u,v}(x)g_{u,v}^{*}(y),
\end{equation*} 
 where $g_{u,v}(x)=\exp(-iu^{\top}x-v^{\top}x)$. Thus Fact~\ref{f:L1} implies that 
 \begin{equation} \label{eq:uv}
 \iL(\kappa_{\PG})(u,v)\geq 0 \quad \text{for all} \quad u,v\in \RR^n.
\end{equation}  
 We want to calculate $\iL(\kappa_{\PG})$ by differentiating under the integral sign of $\iL(\kappa_G)$. We define the differential operators $D_i$ for all $1\leq i\leq 2n$ such that
 \begin{equation*}
D_i= \begin{cases}
\frac 12\left(i \frac{\partial}{\partial u_i}-\frac{\partial}{\partial v_i}\right) & \textrm{if } 1\leq i\leq n,   \\
-\frac {1}{2}\left(i \frac{\partial}{\partial u_i}+\frac{\partial}{\partial v_i}\right) & \textrm{if } n+1\leq i\leq 2n.  
 \end{cases}
 \end{equation*}
For any polynomial $R=R(x,y)$ and $1\leq i\leq n$ by Lemma~\ref{l:uv} we obtain that
\begin{equation} \label{eq:Di} D_i \iL(R\kappa_G)=\iL(x_iR\kappa_G) \quad 
\text{and} \quad D_{n+i}\iL(R\kappa_G)=\iL(y_iR\kappa_G).
\end{equation} 
For the polynomial $P$ choose the set $\iI\subset \NN^{2n}$ such that 
 \begin{equation*}
 P(x,y)=\sum_{(k_1,\dots,k_{2n})\in \iI} a(k_1,\dots, k_{2n})x_1^{k_1}\cdots x_n^{k_n} y_1^{k_{n+1}}\cdots y_n^{k_{2n}}. 
 \end{equation*}
Let us define the differential operator $P(D)$ as 
\begin{equation*}
P(D)=\sum_{(k_1,\dots,k_{2n})\in \iI} a(k_1,\dots, k_{2n})D_1^{k_1}\cdots D_{2n}^{k_{2n}}.
  \end{equation*}
Then \eqref{eq:Di} implies that 
\begin{equation} \label{eq:PD} 
\iL(\kappa_{\PG})=P(D) \iL(\kappa_{G}).
\end{equation}
Clearly $(D_1^{k_1}\cdots D_{2n}^{k_{2n}})\iL(\kappa_{G})=Q_{k_1\dots k_{2n}}\iL(\kappa_{G})$, where $Q_{k_1\dots k_{2n}}$ is a polynomial of degree exactly $k_1+\dots+k_{2n}$. Therefore, by \eqref{eq:PD} we obtain that there is a polynomial $Q$ such that $\deg Q\leq \deg P$ and 
\begin{equation} \label{eq:Q} 
\iL(\kappa_{\PG})=Q\iL(\kappa_{G}).
\end{equation}  
As $\iL(\kappa_{\PG})$ takes only real values, so does $Q$, hence the coefficients of $Q$ are reals by Fact~\ref{f:Q}\,\eqref{i:real}. Now it is enough to show that $\deg Q=\deg P$. Indeed, assume that $\deg Q=\deg P$. Then by Fact~\ref{f:Q}\,\eqref{i:neg} the odd degree polynomial $Q$ takes negative values, so  $\iL(\kappa_{\PG})$ takes negative values by \eqref{eq:Q} as well, but this  contradicts \eqref{eq:uv}. Note that it is unclear yet why $\deg Q<\deg P$ cannot happen due to cancellations in the linear combination of the polynomials $Q_{k_1\dots k_{2n}}$, $(k_1,\dots, k_{2n})\in \iI$.

Finally, we prove $\deg Q=\deg P$. Assume to the contrary that $\deg Q<\deg P$. Observe that \eqref{eq:PG} provides an easy connection between the transforms $\iF$ and $\iL$. Applying this twice with \eqref{eq:Q} implies that 
\begin{align} \label{eq:F}
\begin{split}
\iF(\kappa_{\PG})(u,v)&=\iL(\kappa_{\PG})\left(\frac 12(u+v),\frac i2(v-u)\right) \\
&=Q\left(\frac 12(u+v),\frac i2(v-u)\right) \iL(\kappa_{G})\left(\frac 12(u+v),\frac i2(v-u)\right) \\
&=Q'(u,v)\iF(\kappa_{G})(u,v),
\end{split} 
\end{align} 
where $Q'(u,v)=Q((u+v)/2,i(v-u)/2)$, so $Q'$ is a polynomial  
with $\deg Q'\leq \deg Q$. Now we will take the Fourier transform of both sides of the \eqref{eq:F}. 
On one hand, by the Fourier inversion formula \cite[Theorem~2.4~in~Chapter~6]{steinfourier} we have an absolute constant $c_n>0$ such that
\begin{equation} \label{eq:F1} \iF(\iF(\kappa_{\PG}))(x,y)=c_n \kappa_{\PG}(-x,-y)=c_n P(-x,-y) \kappa_G(-x,-y).
\end{equation} 
On the other hand, differentiating under the integral sign similarly as earlier and using the above Fourier inversion formula, we obtain that there is a polynomial $Q''$ with $\deg Q''\leq \deg Q'$ satisfying
\begin{equation} \label{eq:F2}
\iF(Q'\iF(\kappa_{G}))(x,y)=Q''(x,y)\iF(\iF(\kappa_{G}))(x,y)=c_nQ''(x,y)\kappa_{G}(-x,-y).
\end{equation} 
By \eqref{eq:F} the Fourier transforms in \eqref{eq:F1} and \eqref{eq:F2} are equal, which yields that $Q''(x,y)=P(-x,-y)$ for all $x,y\in \RR^n$. However, $\deg Q''\leq \deg Q<\deg P$ by our indirect hypothesis, which is clearly a contradiction. This completes the proof. 
\end{proof}

\section{A useful preorder on Gaussian kernels} \label{s:order}

Before proving Theorem~\ref{t:order} we need a couple of facts. 

\begin{fact}\label{f:mu} Let $\mu,\nu \in L^2(\RR^{2n})$ be self-adjoint kernels such that $\widehat{\mu}$ and $\widehat{\nu}$ are positive semidefinite. Then $\mu \nu \in L^2(\RR^{2n})$ and $\widehat{\mu \nu}$ is positive semidefinite, too. 
\end{fact} 

\begin{proof}
Clearly $\mu \nu \in L^2(\RR^{2n})$. By the spectral theorem \cite[Theorem~6.2]{Stein} the operator $\widehat{\mu}$ has nonnegative eigenvalues $\{s_i\}_{i\geq 0}$ with eigenfunctions $\{f_i\}_{i\geq 0}$, and similarly $\widehat{\nu}$ has nonnegative eigenvalues $\{t_j\}_{j\geq 0}$ with eigenfunctions $\{g_j\}_{j\geq 0}$. This easily implies that the kernels $\mu,\nu$ can be written as 
\begin{equation*} 
\mu(x,y)=\sum_{i\geq 0} s_i f_i(x) f_i^{*}(y) \quad \text{and} \quad \nu(x,y)=\sum_{j\geq 0} t_j g_j(x) g_j^{*}(y), 
\end{equation*}
where the sums converge in $L^2(\RR^{2n})$. Then clearly 
\begin{equation*} 
(\mu \nu)(x,y)=\sum_{k=0}^{\infty}\sum_{i+j=k} s_it_j (f_ig_j)(x) (f_ig_j)^{*}(y), 
\end{equation*}
where the sum converges in $L^2(\RR^{2n})$. As the kernels $(x,y)\mapsto (f_ig_j)(x) (f_ig_j)^{*}(y)$ define positive semidefinite operators and $s_it_j\geq 0$ for all $i,j\geq 0$, we obtain that $\widehat{\mu \nu}$ is positive semidefinite.
\end{proof}

\begin{fact} \label{f:positive} 
Let $\tau\in L^2(\RR^{2n})$ be a kernel, and let $g\colon \RR^n\to \CC$ be a Lebesgue measurable function. Define $\kappa\colon \RR^{2n}\to \CC$ as 
\begin{equation*} \kappa(x,y)=\tau(x,y)g(x)g^{*}(y).
\end{equation*} 
If $\widehat{\tau}$ is positive semidefinite and $\kappa\in L^2(\RR^{2n})$, then $\widehat{\kappa}$ is also positive semidefinite.  
\end{fact}

\begin{proof}
Define $h\colon \RR^n\to \RR$ as $h(x)=\exp(-|g(x)|)$. Then $h$ is Lebesgue measurable, and $0<|h(x)|\leq 1$ for all $x\in \RR^n$. Let
\begin{equation*} 
\iF=\left\{P(x)h(x)\exp(-|x|^2): \text{ where $P\colon \RR^n\to \CC$ is a polynomial}  \right\},
\end{equation*}
then \cite[Corollary~14.24 or 14.7 Exercise~10]{Schmudgen} yields that $\iF$ is dense in $L^2(\RR^n)$. 

Now we show that $\langle \widehat{\kappa} f, f\rangle\geq 0$ for all $f\in \iF$. Fix an arbitrary $f\in \iF$. As $hg$ is bounded, clearly $fg\in L^2(\RR^n)$. Then $\widehat{\tau}\geq 0$ yields $\langle \widehat{\kappa} f,f \rangle=\langle \widehat{\tau} (fg),fg \rangle\geq 0$.

Finally, we prove that $\langle \widehat{\kappa} f,f\rangle\geq 0$ for all $f\in L^{2}(\RR^n)$. Since $\iF$ is dense in $L^2(\RR^n)$, there is a sequence $f_i\in \iF$ such that $f_i\to f$ in $L^2(\RR^n)$. Clearly $\langle \widehat{\kappa} f_i,f_i\rangle\to \langle  \widehat{\kappa} f, f\rangle$, and we already showed that $\langle \widehat{\kappa} f_i, f_i\rangle\geq 0$, thus we obtain $\langle  \widehat{\kappa} f, f\rangle\geq 0$.
\end{proof}

Now we are ready to prove Theorem~\ref{t:order}.

\begin{order} 
For $(A_0, B_0,C_0), (A_1, B_1,C_1)\in \iG(n)$ the following are equivalent: 
\begin{enumerate} 
\item  \label{i:1} $(A_0, B_0,C_0)\preceq (A_1,B_1, C_1)$;
\item \label{i:2} for each self-adjoint $\sigma\colon \RR^{2n}\to \CC$ if $\sigma \theta(A_i,B_i,C_i)\in L^2(\RR^{2n})$ for $i\in\{0,1\}$ and 
$\widehat{\sigma \theta}(A_0,B_0,C_0)\geq 0$, then $\widehat{\sigma \theta}(A_1,B_1,C_1)\geq 0$;
\item \label{i:3} for every self-adjoint Gaussian kernel $\sigma$ if $\sigma \theta(A_i,B_i,C_i)\in L^2(\RR^{2n})$ for $i\in\{0,1\}$ and $\widehat{\sigma \theta}(A_0,B_0,C_0)\geq 0$, then $\widehat{\sigma \theta}(A_1,B_1,C_1)\geq 0$.
\end{enumerate} 
\end{order}

\begin{proof}
The implication $\eqref{i:2} \Rightarrow \eqref{i:3}$ is straightforward. 

First we prove that $\eqref{i:3} \Rightarrow \eqref{i:1}$. Choose $r\geq 0$ such that the matrices $-A_0+rI_n$ and $-C_0+rI_n$ are positive definite, and define 
\begin{equation} \label{eq:-A} 
\sigma=\theta(-A_0+rI_n, -B_0, -C_0+rI_n)\in L^2(\RR^{2n}).
\end{equation}  
Then $\sigma \theta(A_0,B_0,C_0)=\theta(rI_n, 0, rI_n)\in L^2(\RR^{2n})$, so clearly $\widehat{\sigma \theta}(A_0,B_0,C_0)$ is positive semidefinite. Furthermore, \eqref{eq:-A} implies that 
\begin{equation*} \sigma \theta(A_1,B_1,C_1)=\theta(A_1-A_0+rI_n, B_1-B_0, C_1-C_0+rI_n)\in L^2(\RR^{2n}).
\end{equation*} 

Therefore, \eqref{i:3} yields that 
\begin{equation*} \widehat{\sigma \theta}(A_1,B_1,C_1)=\widehat{\theta}(A_1-A_0+rI_n, B_1-B_0, C_1-C_0+rI_n)\geq 0,
\end{equation*}
which implies \eqref{i:1} by definition. 

Finally, we prove the implication $\eqref{i:1} \Rightarrow \eqref{i:2}$. Let $\sigma\colon \RR^{2n}\to \CC$ be a self-adjoint function such that $\sigma \theta(A_i,B_i,C_i)\in L^2(\RR^{2n})$ for $i\in \{0,1\}$ and 
$\widehat{\sigma \theta}(A_0,B_0,C_0)$ is positive semidefinite. Choose $r\geq 0$ such that 
\begin{equation*} \widehat{\theta}(A_1-A_0+rI_n, B_1-B_0, C_1-C_0+rI_n)\geq 0.
\end{equation*} 
Let us define the kernels $\mu,\nu,\kappa\in L^2(\RR^{2n})$ as
\begin{align*} 
\mu&=\sigma \theta (A_0,B_0,C_0), \\ 
\nu&=\theta(A_1-A_0+rI_n, B_1-B_0, C_1-C_0+rI_n), \\  \kappa&=\sigma\theta(A_1,B_1,C_1). 
\end{align*}
Then we have 
\begin{equation} \label{eq:munu}
\kappa(x,y)=(\mu \nu)(x,y) g(x)g^{*}(y),
\end{equation}
where 
\begin{equation*} 
g(x)=\exp\left(r |x|^2\right). 
\end{equation*} 
We need to prove that $\widehat{\kappa}\geq 0$. As $\widehat{\mu}\geq 0$ and $\widehat{\nu}\geq 0$, Fact~\ref{f:mu} implies that $\widehat{\mu \nu}\geq 0$. Then \eqref{eq:munu} and Fact~\ref{f:positive} with $\tau=\mu \nu$ yield that $\widehat{\kappa}\geq 0$. This completes the proof.
\end{proof}

Now we are ready to prove Claim~\ref{cl:I0I}.   

\begin{I0I} We have $(I_n,0,I_n)\preceq (A,B,C)$ for all $(A,B,C)\in \iG^{+}(n)$. Hence for a self-adjoint polynomial $P\colon \RR^{2n}\to \CC$ the following are equivalent:  
\begin{enumerate} 
\item \label{i:Gauss2} $\widehat{P\theta}(A,B,C)\geq 0$ for all $(A,B,C)\in \iG^{+}(n)$;
\item  \label{i:In} 
$\widehat{P\theta}(I_n,0,I_n)\geq 0$;
\item \label{i:Mercer2} $\sum_{i,j=1}^k c_i c^{*}_j P(x_i,x_j)\geq 0$ for all $x_1,\dots, x_k \in \RR^n$ and $c_1,\dots, c_k\in \CC$.
\end{enumerate} 
\end{I0I}

\begin{proof} First we show that $(I_n,0,I_n)\preceq (A,B,C)$ for all $(A,B,C)\in \iG^{+}(n)$. Indeed, let $(A,B,C)\in \iG^{+}(n)$ and $r=1$. Then 
\begin{equation*} (A-I_n+rI_n,B,C-I_n+rI_n)=(A,B,C)\in \iG^{+}(n),
\end{equation*} 
which implies $(I_n,0,I_n)\preceq (A,B,C)$ by definition. 

Now we prove the equivalences. Clearly $\eqref{i:Gauss2} \Rightarrow \eqref{i:In}$, and Theorem~\ref{t:order} yields that $\eqref{i:In}\Rightarrow \eqref{i:Gauss}$. Since $\theta=\theta(I_n,0,I_n)$ satisfies $\theta(x,y)=E(x)E^{*}(y)$ such that $E(x)=\exp(-|x|^2)$, we easily obtain that $\eqref{i:In}\Leftrightarrow \eqref{i:Mercer}$. 
\end{proof}

\begin{notation} For $U,V\in \RR^{n\times n}$ we write $U>0$ if $U$ is positive definite, $U\geq 0$ if $U$ is positive semidefinite, and $U\geq V$ if $U-V\geq 0$. 
\end{notation}

\begin{claim} \label{cl:AC} Let $A,B,C\in \RR^{n\times n}$ such that $A,C$ are symmetric and positive definite.  
\begin{enumerate} 
\item \label{i:ABC} If $(A,B,C)\in \iG^+(n)$ then $A\geq C$;
\item \label{i:A0C} if $A\geq C$ and $B$ is symmetric then $(A,B,C)\in \iG^+(n)$. 
\end{enumerate}
\end{claim}

\begin{proof} First we prove \eqref{i:ABC}. Applying Theorem~\ref{t:Mercer} for $\theta=\theta(A,B,C)$ with $k=2$ and $c_1=1,c_2=-1$ for vectors $x$ and $-x$ we obtain 
\begin{equation*}
 \theta(x,x)+\theta(-x,-x)\geq \theta(x,-x)+\theta(-x,x),   
\end{equation*}
which easily implies that 
\begin{equation*}
 2\exp\big(-4x^{\top}Cx\big)\geq  2\exp\big(-4x^{\top}Ax\big).
\end{equation*}
Thus for all $x\in \RR^n$ we have $x^{\top}(A-C)x\geq 0$, so $A-C\geq 0$, that is, $A\geq C$.

Now we prove \eqref{i:A0C}. Note that $\theta(A,B,C)(x,y)$ only differs from $\theta(A,0,C)(x,y)$ by factors of the form 
$f_{i,j}(x)f_{i,j}^{*}(y)$, where 
\begin{equation*} 
f_{i,j}(x)=\exp(-2iB_{ij}x_ix_j)\quad \text{for all } 1\leq i<j\leq n. 
\end{equation*}
Since these factors do not change the spectrum, we obtain that $\theta(A,B,C)\in \iG^+(n)$ if and only if $\theta(A,0,C)\in \iG^+(n)$. Hence we need to prove that $\theta(A,0,C)\in \iG^+(n)$. By Claim~\ref{c:positivity} this is equivalent to the fact that all positive eigenvalues of $M=iG\Omega$ are at most $1$. We can calculate the characteristic polynomial as
\begin{equation*}
P(\lambda)=\det\left(M-\lambda I_{2n}\right)=\det\left(\lambda^2 I_n-A^{-1}C\right).
\end{equation*}
Fix $\lambda>1$, it is enough to show that the $P(\lambda)\neq 0$, for which it is enough to prove that $\lambda^2 I_n-A^{-1}C$ is positive definite. As $A\geq C$, we obtain that $I_n\geq A^{-1}C$, so $\lambda^2 I_n-A^{-1}C=(\lambda^2-1)I_{n}+(I_n-A^{-1}C)$ is the sum of a positive definite and a positive semidefinite matrix, hence it is positive definite. 
\end{proof}
Claim~\ref{cl:AC} and the definition of $\preceq$ immediately imply Claim~\ref{cl:prec}.

\begin{C:prec} 
Let $(A_i,B_i,C_i)\in \iG(n)$, $i\in \{0,1\}$. Then 
\begin{enumerate}
\item \label{i:diff1} $(A_0,B_0,C_0)\preceq (A_1,B_1,C_1)$ $\Rightarrow$ $A_1-C_1\geq A_0-C_0$,
\item \label{i:diff2} $A_1-C_1\geq A_0-C_0$ and $B_1-B_0$ is symmetric $\Rightarrow$ $(A_0,B_0,C_0)\preceq (A_1,B_1,C_1)$.
\end{enumerate}
\end{C:prec} 

In order to prove Theorem~\ref{t:equiv} we need the following fact. 

\begin{fact} \label{f:IB} If $(I_n,B,I_n)\in \iG^+(n)$, then $B$ is symmetric. 
\end{fact} 
\begin{proof}
Assume that $(I_n,B,I_n)\in \iG^+(n)$. By Claim~\ref{c:positivity} all positive eigenvalues of $M=iG\Omega$ are at most $1$. We can calculate the characteristic polynomial of $M$ as
\begin{equation*}
P(\lambda)=\det\left(M-\lambda I_{2n}\right)=\det\left[(\lambda^2-1)I_n - \lambda \frac i2\left(B-B^{\top}\right)\right].
\end{equation*} 
It is clear that the constant term of $P$ is $P(0)=\det(-I_n)=(-1)^n$, so by Vieta's formula all roots of $P$ are $\pm 1$, both with multiplicity $n$ according to Theorem~\ref{t:Williamson}. This means that $P(\lambda)=(\lambda^2-1)^n$. Let $L$ be the coefficient of $\lambda^2$ in $P(\lambda)$ minus the coefficient of $\lambda^2$ in $(\lambda^2-1)^n$; clearly $L$ must be $0$. Now we will calculate $L$ in another way. Let $N(\lambda)=(\lambda^2-1)I_n - \lambda (i/2)\left(B-B^{\top}\right)$. In the determinant $P(\lambda)=\det (N(\lambda))$ if the term $\lambda^2$ comes only from the main diagonal, then it has the same coefficient as the coefficient of $\lambda^2$ in $(\lambda^2-1)^n$. As every non-diagonal element of $N(\lambda)$ is a constant multiple of $\lambda$, our only other option is to choose $-1$ from the main diagonal $n-2$ times, and if the remaining rows and columns are of number $j$ and $k$, then its contribution to $L$ is
\begin{equation*} 
(-1)^{n-2}\left[-\frac i2(B_{jk}-B_{kj})\frac i2(B_{kj}-B_{jk})\right]=\frac{(-1)^{n-1}}{4}\left(B_{jk}-B_{kj}\right)^2.
\end{equation*} 
We need to sum these for all pairs $1\leq k<j\leq n$, which implies that
\begin{equation*}
L=\frac{(-1)^{n-1}}{4}\sum_{1\leq k<j\leq n} \left(B_{jk}-B_{kj}\right)^2.
\end{equation*}
As $L=0$, we have $B_{jk}=B_{kj}$ for all $1\leq k<j\leq n$, so $B$ is symmetric. 
\end{proof}

Now we can prove Theorem~\ref{t:equiv}. 

\begin{T:equiv} Let $(A_i,B_i,C_i)\in \iG(n)$, $i\in \{0,1\}$. 
The following are equivalent: 
\begin{enumerate} 
\item \label{i:app} $(A_0,B_0,C_0)\approx (A_1,B_1,C_1)$;
\item \label{i:sym} $A_1-C_1=A_0-C_0$ and $B_1-B_0$ is symmetric. 
\end{enumerate} 
\end{T:equiv}

\begin{proof}[Proof of Theorem~\ref{t:equiv}]
Applying Claim~\ref{cl:prec}\,\eqref{i:diff2} twice yields $\eqref{i:sym} \Rightarrow \eqref{i:app}$.

Now we prove $\eqref{i:app} \Rightarrow \eqref{i:sym}$. By Claim~\ref{cl:prec}\,\eqref{i:diff1} we obtain that $A_1-C_1=A_0-C_0$. Let $N=A_1-A_0=C_1-C_0$ and choose $r\geq 0$ such that 
\begin{equation*} (A_1-A_0+rI_n,B_1-B_0,C_1-C_0+rI_n)=(N+rI_n,B_1-B_0,N+rI_n)\in \iG^+(n).
\end{equation*}
Then $\eqref{i:sym} \Rightarrow \eqref{i:app}$ implies that
\begin{equation*} (N+rI_n,B_1-B_0,N+rI_n) \approx (I_n,B_1-B_0,I_n),
\end{equation*}
so $(I_n,B_1-B_0,I_n)\in \iG^+(n)$. Thus Fact~\ref{f:IB} implies that $B_1-B_0$ is symmetric. 
\end{proof}

\section{Discussion and Conclusions} \label{s:discussion}

The main goal of the paper was to establish practically relevant methods to check the positivity of self-adjoint trace-class integral operators. Our methods provide useful necessary conditions for the entanglement problem due to the Peres--Horodecki criterion, too. Our main examples are taken from the class of polynomial Gaussian operators. This family appears frequently in physics, in the Introduction we gave a short account of the relevant physical models, see Subsection~\ref{ss:motivation}. 

We obtained an important positivity test for polynomial Gaussian operators in Theorem~\ref{Th:BHCS}: we check the positivity of the Gaussian part only, if it fails to be positive, then the corresponding polynomial Gaussian operator cannot be positive either. Deciding the positivity of a Gaussian operator is much easier: it is enough to calculate the symplectic eigenvalues of the Gaussian covariance matrix, recall Claim~\ref{c:positivity}. 

We proved in Theorems~\ref{t:odd} and \ref{t:reducible} that polynomial Gaussian operators with polynomials of odd degree, or more generally reducible to odd degree, cannot be positive semidefinite. 

We introduced a new preorder $\preceq$ on Gaussian kernels such that if $\kappa_{G_0}\preceq \kappa_{G_1}$ and $\sigma\colon \RR^{2n}\to \CC$ is a self-adjoint function such that $\sigma \kappa_{G_i}$ are kernels, then the positivity of $\widehat{\sigma \kappa}_{G_0}$ implies the positivity of $\widehat{\sigma \kappa}_{G_1}$, see Theorem~\ref{t:order}. Note that the validity of this result is not restricted to polynomial Gaussian operators. The essence of this method is that determining the positivity of an integral operator decides the positivity of lots of another operators simultaneously; we can switch the Gaussian parts according to the preorder, see Corollary~\ref{c:two}. In Claim~\ref{cl:I0I} we characterize the polynomials $P$ which always define a positive operator whenever multiplied by a positive Gaussian kernel, see Example~\ref{ex}. We also describe the equivalence classes of Gaussian operators which are equal with respect to our preorder, see Theorem~\ref{t:equiv}. The main application of our preorder is the following: if we want to decide the positivity of $\widehat{\sigma \kappa}_{G_0}$, then we can test a much larger class of operators $\{ \widehat{\sigma \kappa}_G: \kappa_{G_0}\preceq \kappa_G\}$, because all of them must be positive. This might improve any given positivity test. We demonstrated that this method works in practice at the end of Subsection~\ref{ss:preorder} by upgrading the test coming from \cite{Newton}. 
 
Our main results can be translated to the detection of NPT-type entanglement, too, which is a particularly important task in the quantum physics of open systems: any sharper test can be directly applied to entanglement problems. Our methods directly imply the entanglement tests given in Corollaries \ref{c:order}--\ref{c:NPT2}.

Finally, we raised an open question in Problem~\ref{p:problem}, which asks whether the Werner--Wolf theorem can be generalized to polynomial Gaussian operators. If the answer were positive, it would reduce the bipartite entanglement problem for polynomial Gaussian operators to a positivity problem, for which there are a lot more available methods. 

\subsection*{Acknowledgments} We are indebted to J\'ozsef Zsolt Bern\'ad, Andr\'as Bodor, \\ Andr\'as Frigyik, Tam\'as Kiss, M\'aty\'as Koniorczyk, L\'aszl\'o Lisztes, Mikl\'os Pint\'er, Nikolett N\'emet, G\'eza T\'oth, and Andr\'as Vukics for some illuminating conversations.

\subsection*{Funding} The first author was supported by the National Research, Development and Innovation Office -- NKFIH, grants no.~124749 and 146922, and by the J\'anos Bolyai Research Scholarship of the Hungarian Academy of Sciences. \\
The second author was supported by the Hungarian National Research, Development and Innovation Office within the Quantum Information National Laboratory of Hungary grants no.~2022-2.1.1-NL-2022-00004 and 134437.\\
The third author thanks the ”Frontline” Research Excellence Programme of the NKFIH (Grant no.~KKP133827) and Project no.~TKP2021-NVA-04, which has been implemented with the support provided by the Ministry of Innovation and Technology of Hungary from the National Research, Development and Innovation Fund, financed under the TKP2021-NVA funding scheme.

\printbibliography

@article{Salazar,
  title = {Phase-space quantum distributions and information theory},
  author = {Salazar, S.~J.~C. and Laguna, H.~G. and Sagar, R.~P.},
  journal = {Phys.~Rev.~A},
  volume = {107},
  issue = {4},
  pages = {042417},
  numpages = {12},
  year = {2023},
  publisher = {American Physical Society},
  doi = {10.1103/PhysRevA.107.042417},
}

@book{Neumann,
author = {Neumann, J.~v.},
publisher = {Springer},
title = {Mathematische Grundlagen der Quantenmechanik},
year = {1932},
}

@article{Weyl,
author={Weyl, H.},
title={Quantenmechanik und Gruppentheorie},
journal={Z.~Phys.},
year={1927},
volume={46},
pages={1-46},
doi={10.1007/BF02055756},

}

@article{Wigner,
  title = {On the Quantum Correction For Thermodynamic Equilibrium},
  author = {Wigner, E.},
  journal = {Phys.~Rev.},
  volume = {40},
  issue = {5},
  pages = {749--759},
  numpages = {0},
  year = {1932},
  publisher = {American Physical Society},
  doi = {10.1103/PhysRev.40.749},

}

@article{Husimi,
  title={Some Formal Properties of the Density Matrix},
  author={K.~Husimi},
  journal={Phys.~Math. Soc.~Jpn.~3rd Series},
  volume={22},
  issue={4},
  pages={264-314},
  year={1940},
  doi={10.11429/ppmsj1919.22.4_264}
}

@ARTICLE{Moyal,
       author = {Moyal, J.~E. and Bartlett, M.~S.},
        title = "{Quantum mechanics as a statistical theory}",
      journal = {Math.~Proc.~Cambridge Philos.~Soc.},
         year = 1949,
       volume = {45},
       issue = {1},
        pages = {99-124},
          doi = {10.1017/S0305004100000487}
}

@ARTICLE{Hillery,
       author = {Hillery, M. and O'Connell, R.~F. and Scully, M.~O. and Wigner, E.~P.},
        title = "{Distribution functions in physics: Fundamentals}",
      journal = {Phys. Rep.},
         year = 1984,
       volume = {106},
       issue = {3},
        pages = {121-167},
          doi = {10.1016/0370-1573(84)90160-1}
}

@article{Lee,
title = {Theory and application of the quantum phase-space distribution functions},
journal = {Phys.~Rep.},
volume = {259},
issue = {3},
pages = {147-211},
year = {1995},
doi = {https://doi.org/10.1016/0370-1573(95)00007-4},
author = {H.-W.~Lee},
}

@article{Weinbub,
    author = {Weinbub, J. and Ferry, D.~K.},
    title = "{Recent advances in Wigner function \\ approaches}",
    journal = {Appl.~Phys.~Rev.},
    volume = {5},
    pages = {041104},
    year = {2018},
    doi = {10.1063/1.5046663},
    }

@book{nielsen_chuang_2010, 
place={Cambridge}, 
title={Quantum Computation and Quantum Information: 10th Anniversary Edition}, 
DOI={10.1017/CBO9780511976667}, 
publisher={Cambridge University Press}, 
author={Nielsen, M.~A. and Chuang, I.~L.}, 
year={2010}}

@book{Schleich,
  address = {Berlin},
  author = {Schleich, W.~P.},
  publisher = {Wiley-VCH},
  title = {Quantum Optics in Phase Space},
  year = 2001
}

@book{book1,
    author = {Breuer, H.-P. and Petruccione, F.},
    title = "{The Theory of Open Quantum Systems}",
    publisher = {Oxford University Press},
    year = {2007},
    doi = {10.1093/acprof:oso/9780199213900.001.0001},
  
}

@article{HPZ,
  title = {Quantum Brownian motion in a general environment: Exact master equation with nonlocal dissipation and colored noise},
  author = {Hu, B.~L. and Paz, J.~P. and Zhang, Y.},
  journal = {Phys.~Rev.~D},
  volume = {45},
  issue = {8},
  pages = {2843--2861},
  numpages = {0},
  year = {1992},
  publisher = {American Physical Society},
  doi = {10.1103/PhysRevD.45.2843},
 
}

@article{H-Y,
  title = {Alternative derivation of the Hu-Paz-Zhang master equation of quantum Brownian motion},
  author = {Halliwell, J.~J. and Yu, T.},
  journal = {Phys.~Rev.~D},
  volume = {53},
  issue = {4},
  pages = {2012--2019},
  numpages = {0},
  year = {1996},
  publisher = {American Physical Society},
  doi = {10.1103/PhysRevD.53.2012},
  
}

@Article{Gnutzmann,
author={Gnutzmann, S.
and Haake, F.},
title={Positivity violation and initial slips in open systems},
journal={Z.~Phys.~B},
year={1996},
volume={101},
pages={263-273}
}

@article{BLH,
author={Homa, G.
and Bern{\'a}d, J.~Z.
and Lisztes, L.},
title={Positivity violations of the density operator in the Caldeira--Leggett master equation},
journal={Eur.~Phys.~J.~D},
year={2019},
volume={73},
pages={53},
doi={10.1140/epjd/e2019-90604-4}
}

@article{HBCSCS,
  title = {Range of applicability of the Hu-Paz-Zhang master equation},
  author = {Homa, G. and Csord\'as, A. and Csirik, M.~A. and Bern\'ad, J.~Z.},
  journal = {Phys.~Rev.~A},
  volume = {102},
  issue = {2},
  pages = {022206},
  numpages = {13},
  year = {2020},
  publisher = {American Physical Society},
  doi = {10.1103/PhysRevA.102.022206},
 
}

@article{Gosson,
    author = {E.~Cordero and M.~A.~de Gosson and F.~Nicola},
    title = {On the positivity of trace class operators},
    journal = {Adv.~Theor.~Math.~Phys.},
    volume = {23},
    issue = {8},
    year = {2020},
    pages = {2061--2091},
    doi = {10.4310/ATMP.2019.v23.n8.a4},
    publisher = {International Press of Boston}
}

@article{Newton,
doi = {10.1088/1751-8121/acc147},
year = {2023},
publisher = {IOP Publishing},
volume = {56},
issue = {14},
pages = {145203},
author = {G.~Homa and R.~Balka and J.~Z.~Bern\'ad and M.~K\'aroly and A.~Csord\'as},
title = {Newton’s identities and positivity of trace class integral operators},
journal = {J.~Phys.~A: Math.~Theor.}
}

@article{HCSB,
  title = {Analytical evaluation of the coefficients of the Hu-Paz-Zhang master equation: Ohmic spectral density, zero temperature, and consistency check},
  author = {Homa, G. and Bern\'ad, J.~Z. and Csord\'as, A.},
  journal = {Phys.~Rev.~A},
  volume = {108},
  issue = {1},
  pages = {012210},
  numpages = {15},
  year = {2023},
  publisher = {American Physical Society},
  doi = {10.1103/PhysRevA.108.012210}
  
}

@Article{Kastler,
author={Kastler, D.},
title={The $C^{*}$-algebras of a free Boson field},
journal={Comm.~Math.~Phys.},
year={1965},
volume={1},
pages={14-48},
doi={10.1007/BF01649588},
}

@Article{Loupias1,
author={Loupias, G.
and Miracle-Sole, S.},
title={$C^{*}$-alg{\`e}bres des syst{\`e}mes canoniques.~I},
journal={Comm.~Math.~Phys.},
year={1966},
volume={2},
pages={31-48},
doi={10.1007/BF01773339},

}

@Article{Loupias2,
 Author = {Loupias, G. and Miracle-Sole, S.},
 Title = {$C^{*}$-alg{\`e}bres des syst{\`e}mes canoniques.~{II}},
 FJournal = {Annales de l'Institut Henri Poincar{\'e}. Nouvelle S{\'e}rie. Section A. Physique Th{\'e}orique},
 Journal = {Ann.~I.~H.~Poincare~A},
 Volume = {6},
 Pages = {39--58},
 Year = {1967},
 Language = {French},
 zbMATH = {3269808},
 Zbl = {0168.23505}
}

@article{Narkowich1,
  title = {Necessary and sufficient conditions for a phase-space function to be a Wigner distribution},
  author = {Narcowich, F.~J. and O'Connell, R.~F.},
  journal = {Phys.~Rev.~A},
  volume = {34},
  issue = {1},
  pages = {1--6},
  year = {1986},
  publisher = {American Physical Society},
  doi = {10.1103/PhysRevA.34.1}
  
}

@article{Narkowich2,
    author = {Narcowich, F.~J.},
    title = "{Distributions of $\hslash$‐positive type and applications}",
    journal = {J. Math.~Phys.},
    volume = {30},
    pages = {2565-2573},
    year = {1989},
    doi = {10.1063/1.528537}
   }

@article{Werner,
    author = {Bröcker, T. and Werner, R.~F.},
    title = "{Mixed states with positive Wigner functions}",
    journal = {J.~Math.~Phys.},
    volume = {36},
    pages = {62-75},
    year = {1995},
    doi = {10.1063/1.531326},
    }

@article{Luef,
title = {Remarks on the fact that the uncertainty principle does not determine the quantum state},
journal = {Phys.~Lett.~A},
volume = {364},
issue = {6},
pages = {453-457},
year = {2007},
doi = {10.1016/j.physleta.2006.12.024},
author = {M.~A.~de Gosson and F.~Luef},
}

@article{EPR,
title = {Can Quantum-Mechanical Description of Physical Reality Be Considered Complete?},
author = {Einstein, A. and Podolsky, B. and Rosen, N.},
  journal = {Phys.~Rev.},
  volume = {47},
  issue = {10},
  pages = {777--780},
  numpages = {0},
  year = {1935},
  publisher = {American Physical Society},
  doi = {10.1103/PhysRev.47.777},
 }

@book{Bell, place={Cambridge}, edition={2}, title={Speakable and Unspeakable in Quantum Mechanics: Collected Papers on Quantum Philosophy}, DOI={10.1017/CBO9780511815676}, publisher={Cambridge University Press}, author={Bell, J. S. and Aspect, A.}, year={2004}}

@article{HORODECKI19961,
title = {Separability of mixed states: necessary and sufficient conditions},
journal = {Phys.~Lett.~A},
volume = {223},
issue = {1-2},
pages = {1-8},
year = {1996},
doi = {10.1016/S0375-9601(96)00706-2},
author = {M.~Horodecki and P.~Horodecki and R.~Horodecki},
}

@article{SimonR ,
  title = {Peres-Horodecki Separability Criterion for Continuous Variable Systems},
  author = {Simon, R.},
  journal = {Phys.~Rev.~Lett.},
  volume = {84},
  issue = {12},
  pages = {2726--2729},
  numpages = {0},
  year = {2000},
  publisher = {American Physical Society},
  doi = {10.1103/PhysRevLett.84.2726},
  }

@article{GUHNE,
title = {Entanglement detection},
journal = {Phys.~Rep.},
volume = {474},
pages = {1-75},
year = {2009},
doi = {10.1016/j.physrep.2009.02.004},
author = {O.~G\"uhne and G.~T\'oth},
}

@article{Vogel,
  title = {Necessary and sufficient conditions for bipartite entanglement},
  author = {Sperling, J. and Vogel, W.},
  journal = {Phys.~Rev.~A},
  volume = {79},
  issue = {2},
  pages = {022318},
  numpages = {7},
  year = {2009},
  publisher = {American Physical Society},
  doi = {10.1103/PhysRevA.79.022318},
 }

@Article{kina_entanglement,
AUTHOR = {Hsiang, J.-T. and Arısoy, O. and Hu, B.-L.},
TITLE = {Entanglement Dynamics of Coupled Quantum Oscillators in Independent NonMarkovian Baths},
JOURNAL = {Entropy},
VOLUME = {24},
YEAR = {2022},
issue = {12},
pages = {1814},
DOI = {10.3390/e24121814}
}

@article{Miki,
  title = {Non-Gaussian entanglement in gravitating masses: The role of cumulants},
  author = {Miki, D. and Matsumura, A. and Yamamoto, K.},
  journal = {Phys. Rev. D},
  volume = {105},
  issue = {2},
  pages = {026011},
  numpages = {13},
  year = {2022},
  publisher = {American Physical Society},
  doi = {10.1103/PhysRevD.105.026011},
  }

@Article{Duan,
AUTHOR = {Duan, J. and Zhang, L. and Qian, Q. and Fei, S.-M.},
TITLE = {A characterization of maximally entangled two-qubit states},
JOURNAL = {Entropy},
VOLUME = {24},
YEAR = {2022},
issue = {2},
pages = {247},
DOI = {10.3390/e24020247}
}

@book{Murphy,
  title={$C^{*}$-Algebras and Operator Theory},
  author={Murphy, G.~J.},
   year={1990},
  publisher={Elsevier Science \& Technology Books},
doi={10.1016/C2009-0-22289-6}
}

@article{Volume,
doi = {10.1088/1751-8121/ac3469},
year = {2021},
publisher = {IOP Publishing},
volume = {54},
pages = {495302},
author = {A.~Sauer and J.~Z.~Bern\'ad and H.~J.~Moreno and G.~Alber},
title = {Entanglement in bipartite quantum systems: Euclidean volume ratios and detectability by Bell inequalities},
journal = {J.~Phys.~A: Math.~Theor.},
}

@Article{Pirandola,
author={Pirandola, S.
and Mancini, S.},
title={Quantum teleportation with continuous variables: A survey},
journal={Laser Phys.},
year={2006},
volume={16},
pages={1418-1438},
doi={10.1134/S1054660X06100057},
}

@article{Adesso,
title = "Continuous Variable Quantum Information: Gaussian States and Beyond",
author = "G.~Adesso and S.~Ragy and A.~R.~Lee",
year = "2014",
volume = "21",
journal = "Open Syst.~Inf.~Dyn.",
publisher = "World Scientific Publishing",
}

@article{_pereira,
author = {Pereira, J.~L. and Banchi, L. and Pirandola, S.},
title = {Symplectic decomposition from submatrix determinants},
journal = {Proc.~R.~Soc.~A: Math.~Phys.~Eng.~Sci.},
volume = {477},
issue = {2255},
pages = {20210513},
year = {2021},
doi = {10.1098/rspa.2021.0513},


}

@article{Mercer,
author = {Mercer, J.},
title = {XVI.~Functions of positive and negative type, and their connection with the theory of integral equations},
journal = {Philos.~Trans.~R.~Soc.~A},
volume = {209},
pages = {415-446},
year = {1909},
doi = {10.1098/rsta.1909.0016},

}

@article{Lami,
doi = {10.1088/1367-2630/aaa654},
year = {2018},
publisher = {IOP Publishing},
volume = {20},
pages = {023030},
author = {L.~Lami and A.~Serafini and G.~Adesso},
title = {Gaussian entanglement revisited},
journal = {New J.~Phys}

}

@misc{Stewart,
title= {Fredholm, Hilbert, Schmidt: Three Fundamental Papers on Integral Equations },
author= { Translated with commentary by G.~W.~Stewart},
year={2011},
journal={online manuscript},
url={https://users.umiacs.umd.edu/~stewart/FHS.pdf}
}

@book{deGosson_book,
title= {Symplectic Geometry and Quantum Mechanics},
author={de Gosson, M.~A.},
 lccn={2016346359},
  series={Operator Theory:  Advances and Applications},
doi={10.1007/3-7643-7575-2},
   year={ 2006 },
  publisher={Birkh{\"a}user Basel}
}

@book{Folland+1989,

title = {Harmonic Analysis in Phase Space. (AM-122), Volume 122},
author = {G.~B.~Folland},
address = {Princeton},
doi = {10.1515/9781400882427},
year = {1989},
publisher = {Princeton University Press}
}

@Article{Arvind1995,
author={{Arvind}
and Dutta, B.
and Mukunda, N.
and Simon, R.},
title={The real symplectic groups in quantum mechanics and optics},
journal={Pramana},
year={1995},
volume={45},
issue={6},
pages={471-497},
doi={10.1007/BF02848172},

}

@book{Zee2010,
  author = {Zee, A.},
  edition = {Second Edition},
  publisher = {Princeton University Press},
  title = {Quantum field theory in a nutshell},
  year = 2010,
doi={10.1088/0264-9381/28/8/089003}
}

@book{Stein,
author= "Stein, E.~M. and Shakarchi, R.",
title= "{Real analysis: measure theory, integration, and Hilbert
                       spaces}",
publisher= "Princeton University Press",
series        = "Princeton lectures in analysis",
year          = "2005",
}

@article{Open,
  title = {Five Open Problems in Quantum Information Theory},
  journal = {PRX Quantum},
  volume = {3} ,
  issue = {1},
  pages = {010101},
  numpages = {17},
  year = {2022},
  publisher = {American Physical Society},
  doi = {10.1103/PRXQuantum.3.010101},
 author = {Horodecki, P. and Rudnicki, L. and Zyczkowski, K.},
 }

@InProceedings{Gosson_trace_class,
author={Nicola, E.
and de Gosson, M.~A.
and   Nicola,~F.},
editor={Nielsen, Frank
and Barbaresco, Fr{\'e}d{\'e}ric},
title="Quantum Harmonic Analysis and the Positivity of Trace Class Operators; Applications to Quantum Mechanics",
booktitle="Geometric Science of Information",
year="2017",
publisher="Springer International Publishing",
}

@inbook{Hilbert1904,
author="Hilbert, D.",
title="Grundz{\"u}ge einer allgemeinen Theorie der linearen Integralgleichungen",
booktitle="Integralgleichungen und Gleichungen mit unendlich vielen Unbekannten",
year="1989",
publisher="Vieweg+Teubner Verlag",
address="Wiesbaden",
pages="8--171",
doi="10.1007/978-3-322-84410-1_1",
}

@article{Fredholm1903,
  title={Sur une classe d’équations fonctionnelles},
  author={I.~Fredholm},
  journal={Acta Math.},
  volume={27},
  issue={1},
  pages={365-390},
  year={1903},
}

@article{Neumann1927,
author = {Neumann, J.~v.},
journal = {Nachrichten von der Gesellschaft der Wissenschaften zu Göttingen, Mathematisch-Physikalische Klasse},
pages = {245-272},
title = {Wahrscheinlichkeitstheoretischer Aufbau der \\ Quantenmechanik},
volume = {1927},
year = {1927},
}

@article{Perescrit,
  title = {Separability Criterion for Density Matrices},
  author = {Peres, A.},
  journal = {Phys.~Rev.~Lett.},
  volume = {77},
  issue = {8},
  pages = {1413--1415},
  numpages = {0},
  year = {1996},
  publisher = {American Physical Society},
  doi = {10.1103/PhysRevLett.77.1413},
 }

@article{WernerWolf,
  title = {Bound Entangled Gaussian States},
  author = {Werner, R.~F. and Wolf, M.~M.},
  journal = {Phys.~Rev.~Lett.},
  volume = {86},
  issue = {16},
  pages = {3658--3661},
  numpages = {0},
  year = {2001},
  publisher = {American Physical Society},
  doi = {10.1103/PhysRevLett.86.3658},
 }

@book{Gosson_Harmonic,
title={Symplectic Methods in Harmonic Analysis and in Mathematical Physics},
  author={M.~A.~de Gosson},
publisher = {Springer Basel AG},
doi={10.1007/978-3-7643-9992-4},
  year={2011},
 }

@article{Williamson_origin,
 author = {J.~Williamson},
 journal = {Amer.~J.~Math.},
 issue = {1},
 pages = {141--163},
 publisher = {Johns Hopkins University Press},
 title = {On the Algebraic Problem Concerning the Normal Forms of Linear Dynamical Systems},
doi={10.2307/2371062},
volume = {58},
year = {1936}
}

@article{Roux_polygaussian,
title = {Polynomial Gaussian beams and topological charge conservation},
journal = {Opt.~Commun.},
volume = {266},
issue = {2},
pages = {433-437},
year = {2006},
doi = {10.1016/j.optcom.2006.05.038},
author = {F.~S.~Roux},

}

@book{SR,
  title={Methods of Modern Mathematical Physics: Functional analysis (Revised and enlarged edition)},
  author={Reed, M. and Simon, B.},
  lccn={80039580},
  series={Methods of Modern Mathematical Physics},
  year={1980},
  publisher={Academic Press}
}

@Book{Duflo,
author={ Duflo, M.},
title={G{\'e}n{\'e}ralit{\'e}s sur les repr{\'e}sentations induites },
year={1972},
pages = {93-119},
publisher={Groupes de Lie \newline R{\'e}solubles, Monographies de la Soc.~Math.~de France, vol.~4, Dunod, Paris }
}

@ARTICLE{Brislawn,
  title     = "Kernels of Trace Class Operators",
  author    = "Brislawn, C.",
  journal   = "Proc.~Amer.~Math.~Soc.",
  publisher = "American Mathematical Society",
  volume    =  104,
  issue    =  4,
  pages     = "1181--1190",
  year      =  1988,
  doi       = "10.1090/S0002-9939-1988-0929421-X"
}

@book{Serafini,
  title={Quantum Continuous Variables: A Primer of Theoretical Methods (Second Edition)},
  author={Serafini, A.},
  year={2023},
  publisher={CRC Press, Taylor \& Francis Group},
doi={10.1201/9781003250975}
}

@article{SimonWW,
  title = {Gaussian pure states in quantum mechanics and the symplectic group},
  author = {Simon, R. and Sudarshan, E.~C.~G. and Mukunda, N.},
  journal = {Phys.~Rev.~A},
  volume = {37},
  issue = {8},
  pages = {3028--3038},
  numpages = {0},
  year = {1988},
  publisher = {American Physical Society},
  doi = {10.1103/PhysRevA.37.3028}
  
}

@article{Partial_trace,
author = {Dias, N.~C. and de Gosson, M. and Prata, J.~N.},
title = {Partial traces and the geometry of entanglement: Sufficient conditions for the separability of Gaussian states},
journal = {Rev.~Math.~Phys.},
volume = {34},
issue = {3},
pages = {2250005},
year = {2022},
doi = {10.1142/S0129055X22500052},

}

@book{Zyczkowski,
  title={Geometry of Quantum States: An Introduction to Quantum Entanglement},
  author={Bengtsson, I. and Zyczkowski, K.},
  year={2006},
  publisher={Cambridge University Press},
doi={10.1017/CBO9780511535048}
}

@article{Schmidt1907,
author = {Schmidt, E.},
journal = {Mathematische Annalen},
pages = {433-476},
title = {Zur Theorie der linearen und nichtlinearen Integralgleichungen. I. Teil: Entwicklung willkürlicher Funktionen nach Systemen \\ vorgeschriebener},
volume = {63},
year = {1907},
}

@book{schmidt1905entwickelung,
  title={Entwickelung willk{\"u}rlicher Functionen nach Systemen \\ vorgeschriebener},
  author={Schmidt, E.},
  year={1905},
  publisher={Dieterich'sche Univ.-Buchdr.(WF Kaestner)}
}

@article{Infinite,
  title={Infinitely entangled states},
  author={M.~Keyl and D.~Schlingemann and R.~F.~Werner},
 journal={Quantum Inf.~Comput.},
  year={2002},
 pages={281-306},
 volume={3},
 issue={4},
 doi={10.26421/QIC3.4-1}
}

@Article{Coladangelo2020,
author={Coladangelo, A.
and Stark, J.},
title={An inherently infinite-dimensional quantum correlation},
journal={Nat.~Commun.},
year={2020},
volume={11},
pages={3335},
doi={10.1038/s41467-020-17077-9},
}

@Article{Fauseweh2024,
author={Fauseweh, B.},
title={Quantum many-body simulations on digital quantum computers: State-of-the-art and future challenges},
journal={Nat.~Commun.},
year={2024},
volume={15},
pages={2123},
doi={10.1038/s41467-024-46402-9},
}

@Article{Piveteau2022,
author={Piveteau, A.
and Pauwels, J.
and H{\aa}kansson, E.
and Muhammad, S.
and Bourennane, M.
and Tavakoli, A.},
title={Entanglement-assisted quantum communication with simple measurements},
journal={Nat.~Commun.},
year={2022},
volume={13},
pages={7878},
doi={10.1038/s41467-022-33922-5},
}

@Article{Dattoli1996,
author={Dattoli, G.
and Torre, A.
and Lorenzutta, S.
and Maino, G.},
title={Coupled harmonic oscillators, generalized harmonic-oscilla-\\tor eigenstates and coherent states},
journal={Il Nuovo Cimento B (1971-1996)},
year={1996},
volume={111},
pages={811-823},
doi={10.1007/BF02749013}
}

@article{Plenio_2004,
doi = {10.1088/1367-2630/6/1/036},
year = {2004},
volume = {6},
pages = {36},
author = {M.-B.~Plenio and J.~Hartley and J.~Eisert},
title = {Dynamics and manipulation of entanglement in coupled harmonic systems with many degrees of freedom},
journal = {New J.~Phys.},
}

@article{PhysRevA.54.5378,
  title = {Displaced squeezed number states: Position space representation, inner product, and some applications},
  author = {Moller, K.~B. and Jorgensen, T.~G. and Dahl, J.~P.},
  journal = {Phys.~Rev.~A},
  volume = {54},
  issue = {6},
  pages = {5378--5385},
  year = {1996},
  publisher = {American Physical Society},
  doi = {10.1103/PhysRevA.54.5378},
 }

@book{meng2023entangled,
  title={Entangled State Representations in Quantum Optics},
  author={Meng, X.~G. and Wang, J.~S. and Liang, B.~L.},
  year={2023},
  publisher={Springer Nature Singapore},
doi={10.1007/978-981-99-2333-5}
}

@misc{MIPRE,
      title={MIP*=RE}, 
      author={Ji, Z. and Natarajan, A. and Vidick, T. and Wright, J. and Yuen, H.},
      year={2022},
      eprint={2001.04383},
      archivePrefix={arXiv},
      primaryClass={quant-ph}
}

@article{Connes,
 author = {A.~Connes},
 journal = {Ann.~Math.},
 issue = {1},
 pages = {73--115},
 title = {Classification of injective factors Cases $\II_1$, $\II_{\infty}$, $\III_{\lambda}$, $\lambda \neq 1$},
 volume = {104},
 year = {1976},
doi={10.2307/1971057}
}

@article{MIPRE2,
author = {Ji, Z. and Natarajan, A. and Vidick, T. and Wright, J. and Yuen, H.},
title = {MIP* = RE},
year = {2021},
issue_date = {November 2021},
publisher = {Association for Computing Machinery},
address = {New York, NY, USA},
volume = {64},
issue = {11},
doi = {10.1145/3485628},
journal = {Commun.~ACM},
pages = {131–138},
numpages = {8}
}

@Article{Ding2021,
author={Ding, L.
and Mardazad, S.
and Das, S.
and Szalay, S.
and Schollw{\"o}ck, U.
and Zimbor{\'a}s, Z.
and Schilling, C.},
title={Concept of Orbital Entanglement and Correlation in Quantum Chemistry},
journal={J.~Chem.~Theory Comput.},
year={2021},
publisher={American Chemical Society},
volume={17},
issue={1},
pages={79-95},
doi={10.1021/acs.jctc.0c00559},
}

@article{Tsirelson_Bell,
title = "Some results and problems on quantum Bell-type inequalities",
author = "Tsirelson, {B.~S.}",
year = "1993",
volume = "8",
pages = "329-345",
journal = "Hadronic J.~Suppl.",
publisher = "Hadronic Press, Inc.",
issue = "4",
}

@article{Junge,
    author = {Junge, M. and Navascues, M. and Palazuelos, C. and Perez-Garcia, D. and Scholz, V.~B. and Werner, R.~F.},
    title = "{Connes' embedding problem and Tsirelson's problem}",
    journal = {J.~Math.~Phys.},
    volume = {52},
    pages = {012102},
    year = {2011},
    doi = {10.1063/1.3514538},
   }

@book{Schmudgen,
  title={The Moment Problem},
  author={Schm{\"u}dgen, K.},
  series={Graduate Texts in Mathematics},
  year={2017},
  publisher={\\Springer International Publishing}, 
 doi={10.1007/978-3-319-64546-9}
}

@article{Bodri,
  title = {Error-free interconversion of nonlocal boxes},
  author = {Bodor, A. and K\'alm\'an, O. and Koniorczyk, M.},
  journal = {Phys.~Rev.~A},
  volume = {106},
  issue = {1},
  pages = {012223},
  numpages = {6},
  year = {2022},
  publisher = {American Physical Society},
  doi = {10.1103/PhysRevA.106.012223},
  }

@Article{Koniorczyk2024,
author={Koniorczyk, M.
and Naszvadi, P.
and Bodor, A.
and Hanyecz, O.
and Adam, P.
and Pint{\'e}r, M},
title={Implementing no-signaling correlations as a service},
journal={Sci.~Rep.},
year={2024},
volume={14},
pages={10756},
doi={10.1038/s41598-024-59492-8},
}

@book{steinfourier,
  title={Fourier Analysis: An Introduction},
  author={Stein, E.~M. and Shakarchi, R.},
  series={Princeton lectures in analysis},
  year={2003},
  publisher={Princeton University Press}
}

@article{Fahn_2023,
year = {2023},
publisher = {IOP Publishing},
volume = {40},
pages = {094002},
author = {Fahn, M.~J. and  Giesel, K. and  Kobler, M.},
title = {A gravitationally induced decoherence model using Ashtekar variables},
journal ={Class.~Quantum Gravity},
}

@misc{hsiang2024graviton,
title={Graviton physics: Quantum field theory of gravitons, graviton noise and gravitational decoherence -- a concise tutorial}, 
author={Hsiang, J.-T. and  Cho, H.-T. and Hu, B.-L.},
year={2024},
eprint={2405.11790},
archivePrefix={arXiv},
primaryClass={hep-th}
}

@Article{Erratum2019,
author={Homa, G.
and Bern{\'a}d, J.~Z.
and Lisztes, L.},
title={Erratum to: Positivity violations of the density operator in the Caldeira--Leggett master equation},
journal={Eur.~Phys. J.~D},
year={2019},
volume={73},
issue={6},
pages={128},
}

@book{lax2002functional,
  title={Functional Analysis},
  author={Lax, P.~D.},
  lccn={01046547},
  series={Pure and Applied Mathematics: A Wiley Series of Texts, Monographs and Tracts},
  year={2002},
  publisher={Wiley}
}

@article{Wu,
  title = {Quantifying protocol efficiency: A thermodynamic figure of merit for classical and quantum state-transfer protocols},
  author = {Wu, Q. and Ciampini, M.~A. and Paternostro, M. and Carlesso, M.},
  journal = {Phys.~Rev.~Res.},
  volume = {5},
  pages = {023117},
  numpages = {14},
  year = {2023},
  publisher = {American Physical Society},
  }

@misc{Lvovsky,
  title={Production and applications of non-Gaussian quantum states of light},
  author={A.~I.~Lvovsky and P.~Grangier and A.~Ourjoumtsev and V.~Parigi and M.~Sasaki and R.~Tualle-Brouri},
 eprint={2006.16985},
archivePrefix={arXiv},
primaryClass={quant-th},
  year={2020},
  }

@article{Walschaers,
  title = {Non-Gaussian Quantum States and Where to Find Them},
  author = {Walschaers, M.},
  journal = {PRX Quantum},
  volume = {2},
  pages = {030204},
  numpages = {68},
  year = {2021},
  publisher = {American Physical Society},
  }

@article{Remus,
  title = {Damping and decoherence of Fock states in a nanomechanical resonator due to two-level systems},
  author = {Remus, L.~G. and Blencowe, M.~P.},
  journal = {Phys.~Rev.~B},
  volume = {86},
  issue = {20},
  pages = {205419},
  numpages = {14},
  year = {2012},
  publisher = {American Physical Society},
  }

@article{Lukas,
title = {Quantum non-Gaussianity of light and atoms},
journal = {Prog.~Quant.~Electron.},
volume = {83},
pages = {100395},
year = {2022},
author = {L.~Lachman and R.~Filip},
}

@article{Yanagimoto,
author = {R.~Yanagimoto and E.~Ng and A.~Yamamura and T.~Onodera and L.~G.~Wright and M.~Jankowski and M.~M.~Fejer and P.~L.~McMahon and H.~Mabuchi},
journal = {Optica},
issue = {4},
pages = {379--390},
publisher = {Optica Publishing Group},
title = {Onset of non-Gaussian quantum physics in pulsed squeezing with mesoscopic fields},
volume = {9},
year = {2022},
}

@article{Szabo,
  title = {Construction of quantum states of the radiation field by discrete coherent-state superpositions},
  author = {S.~Szabo and P.~Adam and Janszky, J. and Domokos, P.},
  journal = {Phys.~Rev.~A},
  volume = {53},
  pages = {2698--2710},
  year = {1996},
  publisher = {American Physical Society},
  
}

@article{Elliott,
author = {A.~Elliott and S.~Parkins},
journal = {J.~Opt.~Soc.~Am.~B},
issue = {8},
pages = {C53--C67},
publisher = {Optica Publishing Group},
title = {Cavity QED systems for steady-state sources of Wigner-negative light},
volume = {41},
year = {2024},

}

\end{document}